%% file: NewDoFICMIMO.tex
\documentclass[12pt, onecolumn, draftcls]{IEEEtran}
\usepackage{amssymb} 
\usepackage{verbatim}
\usepackage{cite} 
\usepackage{amsmath} 
\allowdisplaybreaks
\usepackage{subfigure} 
\usepackage[dvipsnames]{xcolor}
\usepackage[pdftex, bookmarks, unicode]{hyperref}
\usepackage{url}
\usepackage{tikz}
\usetikzlibrary{shapes}

\newcommand{\dongningstyle}{false}
\input{zymacroset.tex}

\newtheorem{corollary}{Corollary}

\newcommand{\rw}{\RVEC{W}}
\newcommand{\rtw}{\RVEC{\widetilde{W}}}
\newcommand{\ry}{\RVEC{Y}}
\newcommand{\rty}{\RVEC{\widetilde{Y}}}
\newcommand{\rby}{\RVEC{\overline{Y}}}

\newcommand{\rx}{\RVEC{X}}
\newcommand{\rtx}{\RVEC{\widetilde{X}}}
\newcommand{\rhx}{\RVEC{\widehat{X}}}
\newcommand{\rz}{\RVEC{Z}}

\newcommand{\ru}{\RVEC{u}}
\newcommand{\rv}{\RVEC{v}}
\newcommand{\ra}{\RVEC{a}}

\newcommand{\Sqrt}[1]{\sqrt{#1}\,}
\newcommand{\rA}{\RMAT{A}}
\newcommand{\rB}{\RMAT{B}}
\newcommand{\rn}{\RVEC{n}}

\newcommand{\rudx}{\RVEC{\underline{X}}}
\newcommand{\rudy}{\RVEC{\underline{Y}}}
\newcommand{\rudz}{\RVEC{\underline{Z}}}
\newcommand{\rudw}{\RVEC{\underline{W}}}
\newcommand{\rudu}{\RVEC{\underline{u}}}
\newcommand{\rudH}{\RMAT{\underline{H}}}
\newcommand{\rudV}{\RMAT{\underline{V}}}
\newcommand{\rudW}{\RMAT{\underline{W}}}
\newcommand{\rudLam}{\RMAT{\underline{\Lambda}}}

\newcommand{\rudG}{\RMAT{\underline{G}}}
\newcommand{\diag}{\text{diag}}

\newcommand{\rLam}{\RMAT{\Lambda}}
\newcommand{\rSig}{\RMAT{\Sigma}}
\newcommand{\mLam}{\MAT{\Lambda}}
\newcommand{\rU}{\RMAT{U}}
\newcommand{\rW}{\RMAT{W}}
\newcommand{\rV}{\RMAT{V}}
\newcommand{\rtV}{\widetilde{\RMAT{V}}}
\newcommand{\rtv}{\widetilde{\RVEC{v}}}
\newcommand{\rH}{\RMAT{H}}
\newcommand{\rG}{\RMAT{G}}

\newcommand{\rQ}{\RMAT{Q}}

\newcommand{\mU}{\MAT{U}}
\newcommand{\mQ}{\MAT{Q}}
\newcommand{\mR}{\MAT{R}}
\newcommand{\mI}{\MAT{I}}
\newcommand{\mV}{\MAT{V}}
\newcommand{\mW}{\MAT{W}}
\newcommand{\mH}{\MAT{H}}
\newcommand{\mB}{\MAT{B}}
\newcommand{\mA}{\MAT{A}}

\newcommand{\mSig}{\MAT{\Sigma}}

\newcommand{\setD}{\mathcal{D}}

\newcommand{\PP}{\gamma}
\newcommand{\markov}[3]{\ensuremath{#1\text{---}#2\text{---}#3}}
\newcommand{\CN}[2]{\mathcal{CN}(#1, #2)}

\newcommand{\MI}[1]{\ensuremath{\mathcal{I} \left( #1 \right)}}

\newcommand{\cov}[1]{\mathsf{cov}\left\{#1\right\}}

\title{The Degrees of Freedom of MIMO Interference Channels without
  State Information at Transmitters}
\author{Yan Zhu
and Dongning Guo
\thanks{Y.~Zhu was with the Department of Electrical Engineering and
  Computer Science, Northwestern University, Evanston, IL 60208, USA.
  He is now with Broadcom Inc., Sunnyvale, CA. USA.}
\thanks{D.~Guo is with the Department of Electrical Engineering and
  Computer Science, Northwestern University, Evanston, IL, USA.}
\thanks{This work has been presented in part at Allerton Conference on
  Communication, Control and Computing, Monticello, IL, USA in September 2009.}
\thanks{This work was supported by NSF under grant
    CCF-0644344 and DARPA under grant W911NF-07-1-0028.}}

\begin{document}
\maketitle
\begin{abstract}
  This paper fully determines the degree-of-freedom (DoF) region of
  two-user interference channels with arbitrary number of transmit and
  receive antennas in the case of isotropic and independent (or
  block-wise independent) fading, 
  where the channel state information is available to the receivers
  but not to the transmitters. The result characterizes the capacity
  region to the first order of the logarithm of the signal-to-noise 
  ratio~(SNR) in the high-SNR regime. The DoF region is achieved using
  random Gaussian codebooks independent of the channel states, which
  implies that it is impossible to increase the DoF using beamforming 
  and interference alignment in the absence of channel state
  information at the transmitters.
\end{abstract}


\section{Introduction}

The interference channel is one of the most important models for the
physical layer of wireless networks.
Some recent breakthroughs in understanding the
fundamental limits of such channels, with or without multiple
antennas are reported in~\cite{Etkin08IFC,Cadambe08IFA,Shang10MIMOIC, Gou08KIFC-MIMO,
  Annapureddy09MIMOIC}. Most existing studies of interference channels assume that full
channel state information~(CSI) is available to all transmitters and
receivers. In practice, however, the state of the channel is usually measured at
the receivers, and it is often difficult for the transmitters to acquire the CSI
accurately in a timely manner.

This paper studies a two-user multiple-input multiple-output (MIMO)
interference channel subject to isotropic fading, where the channel
state is independent over time, and its realization is known to the
receivers but \emph{not} to the transmitters. The channel model is
described in Section~\ref{sec:model}.  An example of the channel is
illustrated in Fig.~\ref{fig:mimo}.  The degree-of-freedom (DoF)
region of the MIMO interference channel is completely characterized by 
Theorem~\ref{thm:new} in Section~\ref{sec:result}.  This is the main
result in this paper.  The result
indicates that without CSI at the transmitters~(CSIT), no additional
gains in terms of DoF can be achieved using beamforming or interference
alignment, which is in contrast to the results for the case with full
CSI shown in~\cite{Jafar07MIMOIFC}. A detailed proof
Theorem~\ref{thm:new} is developed in Sections~\ref{sec:result}
and~\ref{sec:proof}.  

\begin{figure}
  \centering
  \includegraphics[width=2.6in]{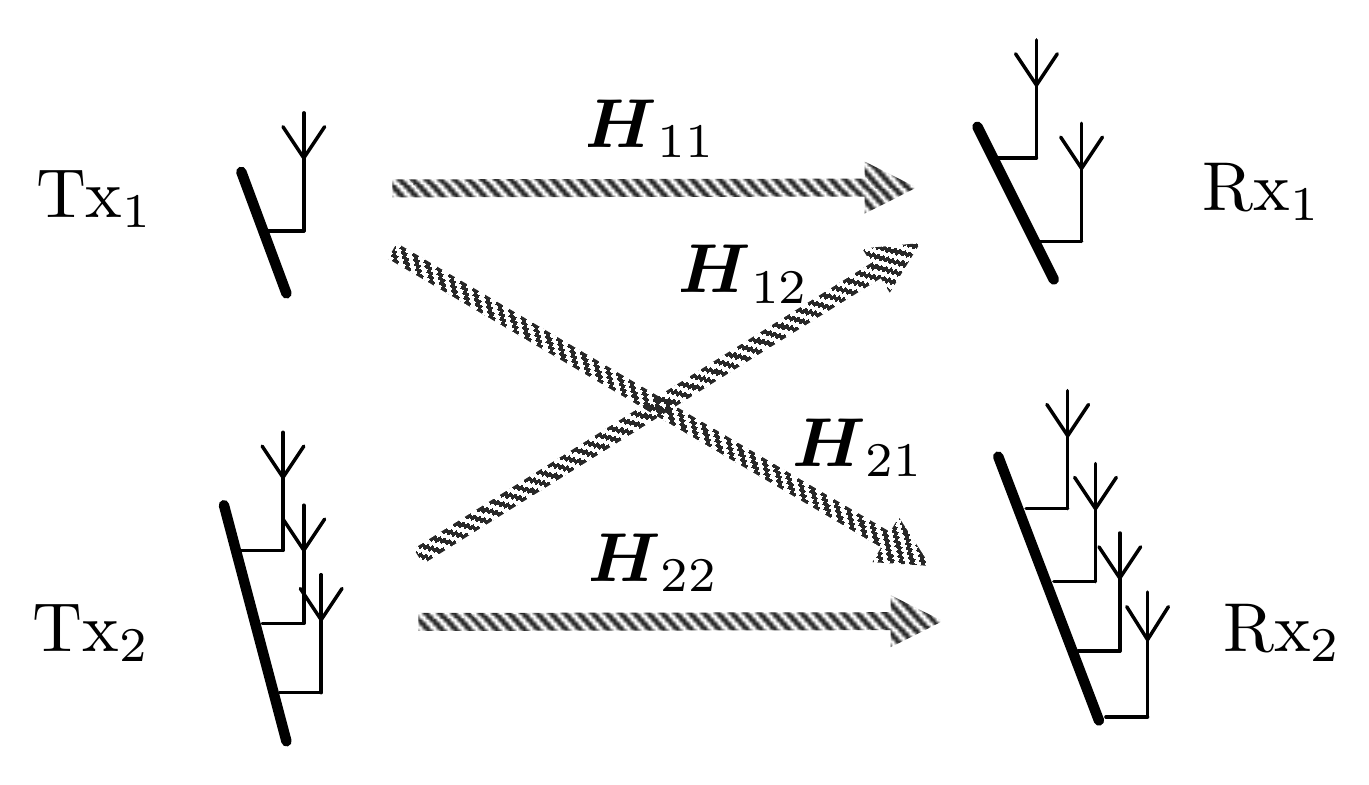}
  \caption{A two-user MIMO interference channel.}
 \label{fig:mimo}
\end{figure}

Related works~\cite{RajPra09IT, Raja09DMT-IFC, Akuiyibo08IFCMIMO,
  Huang09MIMODoFb, Vaze09IFC, ZhuGuo09Allerton} also consider
interference channels without CSIT.
The case of slow fading is modeled as compound interference channels in~\cite{RajPra09IT, Raja09DMT-IFC}, where the capacity of a single-antenna two-user interference channel is studied in~\cite{RajPra09IT}, and the diversity-multiplex trade-off of the same model is studied in~\cite{Raja09DMT-IFC}.
In the case of fast (independent) fading,
Akuiyibo {\it et al} \cite{Akuiyibo08IFCMIMO} derived an outer bound of capacity region for two-user MIMO
interference channels with Rayleigh fading, which is 
tight in terms of the DoF in some special cases.
Tighter outer bounds on the DoF region have been
developed by Huang {\it et al} in~\cite{Huang09MIMODoFb}, who also assume Rayleigh
fading, and by Vaze and Varanasi in~\cite{Vaze09IFC}, who assume a
more general model, and by the authors in~\cite{ZhuGuo09Allerton},
under the assumption of general isotropic fading.\footnote{The fading models of~\cite{Vaze09IFC} and~\cite{ZhuGuo09Allerton} overlap but neither fully covers the other. Both models include independent Rayleigh fading studied in~\cite{Huang09MIMODoFb} as a special case. } A gap remains between the inner and outer bounds in~\cite{Huang09MIMODoFb, Vaze09IFC, ZhuGuo09Allerton}. A specific example is the case where the two users have one and three transmit antennas, and two and four receiver antennas, respectively, as shown in Fig.~\ref{fig:mimo}. The DoF pair $(1,1)$ has been shown to be achievable but the best outer bounds in~\cite{Huang09MIMODoFb, Vaze09IFC, ZhuGuo09Allerton} includes the pair $(1,1.5)$. This paper closes the gap by showing that achievable region of~\cite{ZhuGuo09Allerton} is the exact DoF region. In the aforementioned case, the pair $(1,1.5)$ is not achievable. 

\section{Channel Model}\label{sec:model}
Consider a two-user interference channel, where each transmitter has a
dedicated message for its intended receiver. Suppose transmitter $t$ is equipped
with $M_t$ antennas and receiver $r$ is equipped with $N_r$
antennas for $t, r =1,2$. The signals received
in the $i$-th interval by the two users can be described as:\footnote{As a convention, we use bold fonts to denote random variables, random vectors and random matrices, and we use the corresponding normal fonts to denote their realizations. }
\begin{subequations}\label{eq:sys}
  \begin{align}
    \ry[i] =   \rH_{11}[i]\rw[i] +   \rH_{12}[i]\rx[i] + \ru_1[i] \\
    \rz[i] =   \rH_{21}[i]\rw[i] +   \rH_{22}[i]\rx[i] + \ru_2[i]
  \end{align}
\end{subequations}
where 
$\rw(M_1\times 1)$ and $\rx(M_2 \times 1)$ denote the transmitted
signals, $\rH_{rt}(N_r\times M_t)$ denotes the channel from
transmitter $t$ to receiver $r$, and $\ru_r(N_r\times 1)$ denotes the
thermal noise at receiver $r$, which consists of independent
identically distributed~(i.i.d.) circularly symmetric
complex-Gaussian~(CSCG) random variables of unit variance (denoted by
$\ru_r \sim \CN{0}{\mI_{N_r}}$). The noise process $\{\ru_r[i]\}$ is
i.i.d.\ over time ($i=1,2,\dots$) and independent of the signals and
fading processes $\{\rH_{r1}[i], \rH_{r2}[i]\}$.

The usual power constraint on all codewords of both users is assumed,
\IE, codewords $(w[1], \dots, w[n])$ and $(x[1], \dots, x[n])$ satisfy  
\begin{align}
 \frac{1}{n}\sum_{i=1}^n \|w[i]\|^2 \leq \PP \quad \text{and} \quad
 \frac{1}{n}\sum_{i=1}^n \|x[i]\|^2 \leq \PP  \nonumber
\end{align} 
where $\|\cdot\|$ stands for the Euclidean norm of a vector (more 
generally, it denotes the Frobenius norm of a matrix). 
Since the noise processes are normalized, $\PP$ is regarded as the 
constraint on the average transmit signal-to-noise ratios~(SNR). 
 
The no-CSIT assumption means that the realization of $(\rH_{r1},
\rH_{r2})$ is available to receiver $r$ only ($r=1,2$), whereas the 
transmitters have no knowledge about the channel  
matrices except for their statistics.
The fading process is assumed to be block-wise independent, \IE, the
channel matrices $\rH_{rt}[i]$ remain the same in a constant $T$
consecutive time slots and then change to independent values in the
next block of $T$ slots.  The constant $T$ is often referred to as
the \emph{coherent time}~\cite{Tse05book}. Moreover, the coherence
blocks of all links are perfectly aligned, meaning that the gains
  of all links change at the same time. In particular, if $T=1$, the
  fading process becomes i.i.d. over time. 

The statistics of the fading processes are arbitrary except that all
$\rH_{rt}$ are almost surely of full rank, of finite average power,
\emph{i.e.}, $\MEXP\|\rH_{rt}\|^2< \infty$, and {\em isotropic} in the
following sense:

\begin{defn}
  A complex-valued random matrix $\rG$ is \emph{isotropic} if $\rG
  \mQ$ is identically distributed as $\rG$ for every deterministic
  unitary matrix $\mQ$ of compatible size. 
\end{defn}

We adopt this notion of isotropic fading, which was introduced
in~\cite{Zheng02MIMO}.  
In the absence of CSIT, isotropic fading is a plausible assumption
because there is no reason to prefer signaling toward any  
direction to any other one. Furthermore, many important fading models
belong to this category, including Rayleigh fading studied
in~\cite{Huang09MIMODoFb}, where the channel matrices consist of
i.i.d.\ CSCG entries. 

\section{The Main Theorem and Achievability Proof}\label{sec:result}

A rate pair $(R_1, R_2)$ is said to be achievable if there exist two codebooks
of size $\left\lceil 2^{nR_1} \right\rceil$ and $\left\lceil 2^{nR_2}
\right\rceil$ for the two users, respectively, such that the average
decoding error at each receiver vanishes as the code length
$n\to\infty$. The DoF region is defined as\footnote{Throughout this
  paper, the units of information are bits and all logarithms are of
  base 2.  The DoF is of course invariant to the units of information.}
\begin{multline*}
  \setD = \Big\{(d_1, d_2) \Big| \exists \text{ positive achievable pair }(R_1(\PP),
      R_2(\PP)) \\
      \text{ with }  d_j = \lim_{\PP \to \infty}
      \frac{R_j(\PP)}{\LOG{1+\PP}},  \, j=1,2\Big\}.
\end{multline*}
Evidently, a DoF is essentially the number of single-antenna point-to-point links that 
provides the same rate 
at high SNRs~\cite{Telatar99MIMO, Jafar07MIMOIFC}.\footnote{
The generalized degree of freedom~(GDoF) 
proposed in~\cite{Etkin08IFC}  
is out of the scope of this paper.}

\begin{theorem}\label{thm:new}
  Suppose user~1 has no more receive antennas than user~2, \IE, $N_1 \leq N_2$.
  The DoF region of channel~\eqref{eq:sys} with full rank isotropic fading
   consists of all rate 
   pairs $(d_1,d_2)$ satisfying 
   \begin{subequations}\label{eq:new}
     \begin{align}
       &0 \le d_j \leq \min(M_j, N_j)\,, \quad j=1,2\label{eq:new1} \\
       &d_1 + \frac{\min(M_2, N_1)-L}{\min(M_2, N_2)-L} (d_2-L) \leq \min(M_1, N_1)
       \label{eq:new2}
     \end{align}
   \end{subequations}
   where
   \begin{align}
     L=\min(M_1+M_2, N_1)-\min(M_1, N_1)      
   \end{align}
   and we use the
   convention that $\frac{0}{0}=1$.  The DoF region in the case of
   $N_1\ge N_2$ is similarly determined by symmetry.
 \end{theorem}

The coherent time $T$ has no bearing on the DoF region. The assumption that
  all links have aligned coherent blocks in model~\eqref{eq:sys} is important, as it prohibits
  interference alignment over each coherence block.  In fact, if the direct links and cross links
  have staggered coherence blocks or different block sizes, interference alignment becomes
  possible~\cite{Jafar09BlindIA, KeWan10X}.  This is out of the scope of this paper.

The inequalities~\eqref{eq:new1} are the single-user bounds for the two users.  As we shall
  see, $L$ can be interpreted as the maximum DoF of user~2 without having negative impact on the DoF
  of user~1. Therefore,~\eqref{eq:new2} describes the trade-off between the DoFs of the two users by
  carefully balancing the interference, after $L$ degrees of freedom are guaranteed for user~2.


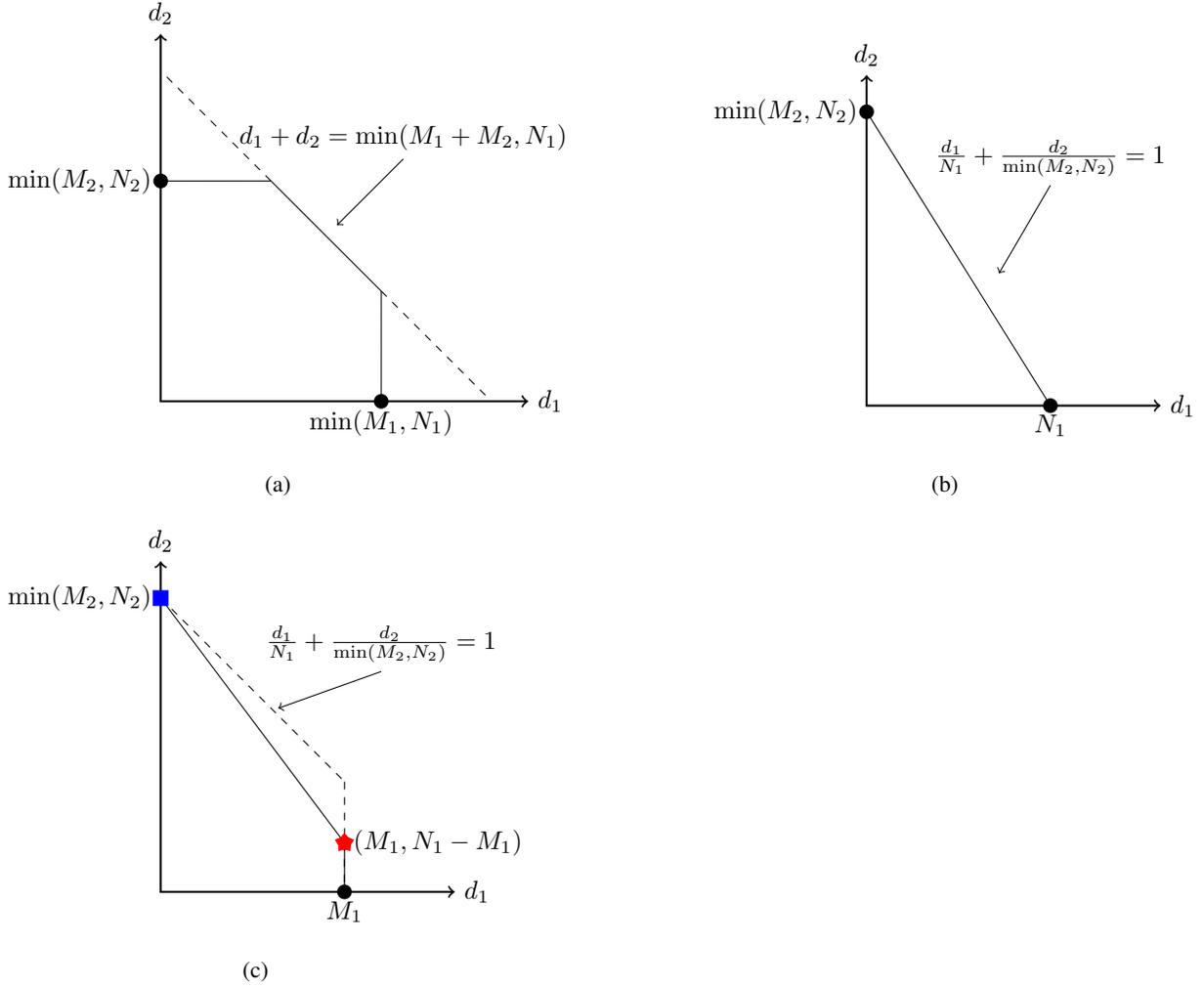
\begin{figure}
  \begin{small}
 \subfigure[]{
    \label{fig:mul1}
    \begin{tikzpicture}
        \draw [<->,thick] (0,5) node (yaxis) [above] {$d_2$}
        |- (5,0) node (xaxis) [right] {$d_1$};
        \coordinate  (r_2) at (0,3);
        \coordinate  (r_22) at (10,3);
        \coordinate  (r_1) at (3,0);
        \coordinate  (r_11) at (3,10);
        \coordinate  (a_2) at (0,4.5);
        \coordinate  (a_1) at (4.5,0);
        \coordinate (a) at (intersection of r_2--r_22 and a_1--a_2);
        \draw (r_2) node[left]{$\min(M_2,N_2)$}-- (a);
        \draw[dashed] (a) -- (a_2);
        \coordinate (b) at (intersection of a_1--a_2 and r_1--r_11); 
        \draw (a) -- (b);
        \draw (b)--(r_1) node[below]{$\min(M_1,N_1)$};
        \draw[dashed] (b)--(a_1);
        \fill[black] (r_2) circle (1mm);
        \fill[black] (r_1) circle (1mm);
        \draw [->] (3.3,3.3) node[above] {$d_1+d_2=\min(M_1+M_2, N_1)$}--(2.4,2.4);
      \end{tikzpicture}}
\hspace{.5cm}
  \subfigure[]{
    \label{fig:mul2}
    \begin{tikzpicture}
        \draw [<->,thick] (0,4.5) node (yaxis) [above] {$d_2$}
        |- (4,0) node (xaxis) [right] {$d_1$};
        \coordinate  (r_2) at (0,4);
        \coordinate  (r_22) at (10,5);
        \coordinate  (r_1) at (2.5,0);
        \coordinate  (r_11) at (2.5,10);
        \coordinate  (a_1) at (3,0);
       \draw (r_2) node[left]{$\min(M_2,N_2)$} -- (r_1) node[below]{$N_1$};
       \fill[black] (r_2) circle (1mm);
        \fill[black] (r_1) circle (1mm);
       \draw [->] (2.5,3) node[above] {$\frac{d_1}{N_1}+\frac{d_2}{\min(M_2,N_2)}=1$}--(1.8,1.8);
      \end{tikzpicture}}
  \subfigure[]{
    \label{fig:mul3}
    \begin{tikzpicture}
        \draw [<->,thick] (0,4.5) node (yaxis) [above] {$d_2$}
        |- (4,0) node (xaxis) [right] {$d_1$};
        \coordinate  (r_2) at (0,4);
        \coordinate  (r_22) at (10,5);
        \coordinate  (r_1) at (2.5,0);
        \coordinate  (r_11) at (2.5,10);
        \coordinate  (a_1) at (3,0);
        \coordinate (b_1) at (4,0);
        \coordinate (a) at (intersection of r_2--a_1 and r_1--r_11);
        \coordinate (b) at (intersection of r_2--b_1 and r_1--r_11);
       \draw (r_2) node[left]{$\min(M_2,N_2)$} -- (a);
        \draw (a) node[right]{$(M_1, N_1-M_1)$}--(r_1) node[below]{$M_1$};
        \draw[dashed] (r_2) -- (b);
        \draw[dashed] (b) -- (r_1);
         \draw [->] (3,3) node[above] {$\frac{d_1}{N_1}+\frac{d_2}{\min(M_2,N_2)}=1$}--(1.6,2.5);
      \node[rectangle, inner sep=1mm,draw=blue, fill=blue] at (r_2) {}; 
       \fill[black] (r_1) circle (1mm);
        \node[star, inner sep=.6mm,draw=red, fill=red] at (a) {};
     \end{tikzpicture}}
\end{small}
  \caption{DoF regions for the cases of
    (a) $N_1\geq 
      M_2$, (b) $M_2>N_1$, $M_1 \geq N_1$, and (c) $M_2>N_1>M_1$. The
      outer bound developed in~\cite{ZhuGuo09Allerton,
        Huang09MIMODoFb, Vaze09IFC} agrees with the exact DoF region
      in cases (a) and (b) but is strictly looser in case (c), where the
      previous outer bound is shown using dashed lines. }
    \label{fig:mul}
\end{figure}

The achievability part of Theorem~\ref{thm:new} can be proved by
further dividing the parameter space (assuming $N_1\le N_2$ without
loss of generality) into the following three cases:
\begin{itemize}
\item[a)] $M_2 \le N_1$.
  In this case~\eqref{eq:new2} becomes
  \begin{align}
    d_1 + d_2 \le \min(M_1+M_2,N_1)\,.
  \end{align}
See Fig.~\ref{fig:mul1} for an illustration.  
The DoF pair $(d_1, d_2)$ falls within the intersections of the DoF
regions of two multiaccess channels~(MAC): one formed by the two
transmitters and receiver~1; and the other formed by the two
transmitters and receiver~2. Therefore, the
DoF region is achievable by letting both users employ
independent random Gaussian codebooks and transmit common
messages only.
Since $N_1 \leq N_2$, receiver~2 can
always decode the message of user~1 in the high SNR regime.   

\item[b)] $M_2 > N_1$ and $M_1\ge N_1$.
  In this case $L=0$ and~\eqref{eq:new2} becomes
  \begin{align}
    \frac{d_1}{N_1} + \frac{d_2}{\min(M_2,N_2)} \le 1\,. \label{eq:caseb}
 \end{align}
  The region becomes a triangle as shown in Fig.~\ref{fig:mul2}.
  Since for both $j=1$ and $j=2$, user $j$ can achieve the single-user
  DoF $\min(M_j, N_j)$ as long as the other user is silent.  It is
  easy to see that the DoF pairs $(N_1,0)$ and $(0,\min(M_2,N_2))$
  are achievable.  Hence the region confined by~\eqref{eq:caseb} 
 can be achieved by time sharing.  

\item[c)] $M_2 > N_1 > M_1$.
  In this case $L=N_1-M_1$ and~\eqref{eq:new2} becomes
  \begin{align}
    \frac{d_1}{M_1} +& \frac{d_2}{\min(M_2,N_2)-N_1+M_1} \nonumber\\
    &\qquad \le \frac{\min(M_2,N_2)}{\min(M_2,N_2)-N_1+M_1}
    \,. \label{eq:casec}
 \end{align}
 The capacity region becomes a trapezoid, as illustrated in
 Fig.~\ref{fig:mul3}.
 It suffices to show the corner points
 on the dominant face of the region are achievable. Evidently, the DoF pair
    $(0,\min(M_2, N_2))$ can be achievable by activating only
    user~2. The pair $(M_1, N_1-M_1)$ is in fact within the
    intersection of DoF regions of the two MAC channels described in
    Case (a), which is evidently achievable.  
\end{itemize}
In all, the achievability part of Theorem~\ref{thm:new} has been established.

Note that for Cases (a) and (b), the DoF region agrees with the previous outer bound developed in~\cite{ZhuGuo09Allerton, Huang09MIMODoFb, Vaze09IFC}. However, for Case (c), the previous outer bound is strictly loose. 

The preceding proof indicates that the DoF region can be achieved
either through time-division multi-access (TDMA) or by the
Han-Kobayashi scheme with common messages
only~\cite{Han81Interference}. It suffices to use random Gaussian
codebooks independent of the fading processes.

\section{Proof of the Converse of Theorem~\ref{thm:new}}\label{sec:proof}

We assume $N_1 \le N_2$ throughout this section. 
We adopt the following notational convention. The sequence $\rx[1], \dots, \rx[n]$ is denoted by $\rx^n$ or $\{\rx\}^n$. For simplicity, let $\rH$ denote $(\rH_{11}, \rH_{12}, \rH_{21}, \rH_{22})$ so that $\rH^n$ denotes all the channel matrices over $n$ time slots.

\subsection{Fading Statistics Revisited}
\label{s:fading}

To facilitate the proof, we shall modify the assumption on the the fading
channel matrices $\rH_{rt}$ in this section without changing the
capacity region.  Roughly speaking, isotropic fading can be decomposed
into two independent components: the ``amplitude'' and the uniformly
distributed ``phase.'' Precisely, we have the following result:
\begin{lemma}\label{lemma:decomp}
  Let $\rG(N\times M)$ be an isotropic random matrix and 
  $K=\min( M, N )$.  Let a compact singular value decomposition
  (SVD) of $\rG$ be
  $\rG = \rW \rLam \rV_1^\dag$ with $\rW (N\times K)$, $\rLam(K \times
  K)$ and $\rV_1 (M \times K)$.  Let $\rQ$ be independent of $\rG$ and 
  uniform distributed on the set of $M\times M$ unitary
  matrices: $\mathcal{Q}=\{\mQ \in \mathbb{C}^{M \times M}: \mQ^\dag\mQ=\mI_M \}$.
  Set $\rV = \rQ \rV_1$.  Then the following properties hold:
 \begin{enumerate}
 \item $\rV_1^\dag \rV_1 = \rV^\dag \rV = \rW^\dag \rW = \mI_K$, and
   $\rLam$ is diagonal with non-negative elements; 
 \item$\rV$ is independent of $(\rW, \rLam, \rV_1)$ and is uniformly
   distributed on $\mathcal{V}=\{\mV \in \mathbb{C}^{M\times K}:
   \mV^\dag\mV=\mI_K \}$; 
 \item $\rG$ and $\rW \rLam \rV^\dag$ are identically distributed,
   denoted by $\rG \sim \rW \rLam \rV^\dag$. 
 \end{enumerate}
\end{lemma}

\begin{proof}
  Property 1 is straightforward by the definition of SVD.  In
  particular, both $\rW$ and $\rV_1$ have orthogonal columns.

 Noting that conditioned on $\rV_1=\mV_1$, $\rV=\rQ\mV_1$ is 
 uniform on $\mathcal{V}$, we conclude that $\rV$ uniform distributed and independent
  of $(\rW,\rLam,\rV_1)$.  Hence Property 2 holds.

  By Definition 1, $\rG$ is identically distributed as $\rG\mQ$,
  which in turn is identically distributed as $\rG\rQ$.  Thus Property
  3 holds, i.e., $\rG\sim\rW\rLam\rV^\dag$.
\end{proof}


The following is a direct consequence of Lemma~\ref{lemma:decomp}:
\begin{corollary}\label{c:decomp}
  Let $(\rG,\rW,\rLam,\rV)$ be defined as in
  Lemma~\ref{lemma:decomp}.  Define block-diagonal matrices
  $\rudG=\text{diag}(\rG, \dots, \rG )$, $\rudW=\text{diag}(\rW,
  \dots, \rW)$, $\rudLam=\text{diag}(\rLam, \dots, \rLam)$ and
  $\rudV=\text{diag}(\rV, \dots, \rV)$, each with $T$ diagonal blocks.
  Then $\rudG \sim \rudW \rudLam \rudV^\dag$.
\end{corollary}

We remark that in general $\rV_1$ is not independent of $(\rW,\rLam)$.
By scrambling $\rV_1$ using uniformly distributed $\rQ$, we obtain
$\rV$, which is guaranteed to be uniformly distributed and independent
of $(\rW,\rLam)$ by Lemma~\ref{lemma:decomp}.

From Lemma~\ref{lemma:decomp}, we can obtain matrices
$(\rW_{rt},\rLam_{rt},\rV_{rt})$ from the compact SVD of $\rH_{rt}$,
which satisfy the three properties given in the lemma.  In particular,
$\rV_{rt}$ is uniformly distributed and independent of $\rH_{rt}$.
For every $r,t=1,2$, channel matrix $\rH_{rt}$ is identically
distributed as $\rW_{rt} \rLam_{rt} \rV_{rt}^\dag$, although they are
not equal in general.  Since the channel capacity depends only on the
statistics of the channel state, we can substitute $\rH_{rt}$ by
$\rW_{rt} \rLam_{rt} \rV_{rt}^\dag$ in model~\eqref{eq:sys} for
$t,r=1,2$ without
changing the capacity region.  This substitution allows a simple proof
of the converse part of Theorem~\ref{thm:new}.  Therefore, with slight
abuse of notation, we let the channel matrices be $\rH_{rt} = \rW_{rt}
\rLam_{rt} \rV_{rt}^\dag$ from this point onward.  Moreover, we let
the decomposition $(\rW_{rt}, \rLam_{rt}, \rV_{rt})$ be determined by $\rH_{rt}$.

\subsection{Preliminary Results}\label{sec:lemmas}

We first develop several preliminary results to facilitate the proof.
The following theorem, proved in Appendix~\ref{app:Gau}, is a simple generalization of \cite[Theorem 3]{Zamir04Gauss} to vector channels. 
\begin{theorem}[Gaussian input is not too bad] \label{thm:Gau}
Suppose that $\rw$ and $\rtw$ are two random $M$-vectors, $\mH(N\times M)$ is a full-rank
deterministic matrix, and $\rv$ is a random $N$-vector which is independent of $\rw$ and $\rtw$.  We assume that $\MEXP\|\rw\|^2\leq \PP$. Then
\begin{align} \label{eq:Gau1} 
  \MI{\mH\rw+ \rv; \rw} \leq \MI{\mH\rtw+\rv; \rtw} + \sup_{\MEXP
    \|\ra\|^2 \leq \PP } \MI{\mH\ra+ \mH\rtw; \ra}.
\end{align}
In particular, if $\rtw$ has distribution $\CN{0}{\frac{\PP}{M}\mI}$, then
\begin{align}
  \MI{\mH\rw+ \rv; \rw} \leq \MI{\mH\rtw+\rv; \rtw} +C^* \label{eq:Gau2} 
\end{align}
where 
\begin{align}
  C^* = \min(M, N)\LOG{1+\frac{M}{\min(M,N)}} \,.\label{eq:def_Cp}
\end{align}
Furthermore, for channel model~\eqref{eq:sys} and 
regarding $\rH_{21}[i]\rx[i]+\ru_1[i]=\rv[i]$,
we have
  \begin{align}
    \MI{\ry^n ; \rw^n | \rH^n} \leq \MI{
      \rty^n; \rtw^n | \rH^n } + nC^* \label{eq:Gau3}
  \end{align} 
where 
\begin{align}
   \rty[i]= \rH_{11}[i]\rtw[i] + \rH_{12}[i]\rx[i] + \ru_1[i]  \label{eq:Gau4}
\end{align}
for $i=1,\dots,n$ and $\rtw[i]\sim \CN{0}{1}$ are i.i.d.\ over time ($i=1,2,\dots$). 
\end{theorem}

The following lemma, shown in Appendix~\ref{app:amp}, puts an upper bound on the change of mutual information due to change of the amplitudes. 

\begin{lemma}\label{lemma:amp}
  Let $\rLam_1$ and $\rLam_2$ be two $M\times M$ diagonal random matrices with
  strictly positive diagonal elements almost surely. Let $\rx$ denote
  a random vector and $\ru$ a CSCG random vector with arbitrary
  covariance, both of dimension $M$. Assume that $\rx$, $\ru$ and
  $(\rLam_1,\rLam_2)$ are 
  independent. Define random matrix $\rLam_{\min}=\min(\rLam_1,
  \rLam_2)$ as the element-wise minimum. Then 
  \begin{align}
    \label{eq:amp}
    I (\rLam_2\rx &+ \ru; \rx | \rLam_2
    ) - I (\rLam_1\rx + \ru; \rx | \rLam_1) \nonumber \\
    &\leq  2\MEXP\LOG{\frac{\det \rLam_2}{\det \rLam_{\min}}} \nonumber \\
    & \leq  2\MEXP \log^+ \det \rLam_2 
+ 2\MEXP \left[
      \log^+\frac{1}{\det\rLam_{\min}} \right]
  \end{align}
  where $\log^+(x)=\log \max(1,  x )$. Evidently, if $\rLam_1$ and
  $\rLam_2$ are deterministic, the inequalities hold with all
  expectations and conditionings dropped. 
\end{lemma} 

\begin{lemma} \label{lemma:phase}
  Let $\rx$ be a random vector in $\mathbb{C}^M$, $\ru_j \sim \CN{0}{\mI_{K_j}}$, $j=1,2,3$, and $K_1 \leq K_2 \leq M$.  In addition, let $\rV_j$ be a random $M \times K_j$ matrix for $j=1,2,3$. Suppose that conditioned on $\rV_3=\mV_3$, $\rV_j$ is uniformly distributed on $\mathcal{V}_j=\{\mV \in \mathbb{C}^{M\times K_j}| \mV^\dag \mV = \mI_{K_j} \text{ and }\mV^\dag\mV_3=0 \}$ for $j=1,2$. Suppose also that $\rx$, $\ru_1$, $\ru_2$, $\ru_3$ and $\rV=(\rV_1, \rV_2, \rV_3)$ are mutually independent. Then
  \begin{align}
    \frac{1}{K_1}\MI{\rV_1^\dag \rx+ \ru_1; \rx \Big| \rV_3^\dag \rx + \ru_3, \rV} \geq \frac{1}{K_2}\MI{\rV_2^\dag \rx+ \ru_2; \rx \Big| \rV_3^\dag \rx + \ru_3, \rV} . \label{eq:phase1}
  \end{align}
Furthermore, suppose $(\rV_1[i], \rV_2[i], \rV_3[i])_{i=1}^n$ is
i.i.d.~following the joint distribution of $(\rV_1, \rV_2, \rV_3)$, then
\begin{align}
  \frac{1}{K_1}\MI{ \{ \rV_1^\dag \rx+ \ru_1 \}^n ; \rx^n \Big| \{ \rV_3^\dag \rx+ \ru_3 \}^n, \rV^n} \geq \frac{1}{K_2}\MI{ \{ \rV_2^\dag \rx+ \ru_2 \}^n; \rx^n \Big| \{ \rV_3^\dag \rx+ \ru_3 \}^n, \rV^n}.  \label{eq:phase2}
\end{align}
In particular, if $\rV_3\equiv 0$,~\eqref{eq:phase1} and~\eqref{eq:phase2} become
\begin{align*}
  \frac{1}{K_1}\MI{\rV_1^\dag \rx+ \ru_1; \rx \Big| \rV_1} \geq \frac{1}{K_2}\MI{\rV_2^\dag \rx+ \ru_2; \rx \Big| \rV_2}
\end{align*}
and 
\begin{align*}
 \frac{\MI{ \{ \rV_1^\dag \rx+ \ru_1 \}^n ; \rx^n \Big| \rV_1^n}}{K_1} \geq \frac{\MI{ \{ \rV_2^\dag \rx+ \ru_2 \}^n; \rx^n \Big| \rV_2^n}}{K_2}
\end{align*} 
respectively.
\end{lemma}

Proved in Appendix~\ref{app:phase}, Lemma~\ref{lemma:phase} essentially states that the mutual information per dimension decreases with the dimensionality of the uniform transformation of the channel input. The following corollary is a simple extension of Lemma~\ref{lemma:phase} to block-diagonal matrices. 
  \begin{corollary}\label{c:phase}
    Suppose that $\rudV_1=\diag(\rV_1, \dots, \rV_1)$, $\rudV_2=\diag(\rV_2, \dots, \rV_2)$, and $\rudV_3=\diag(\rV_3, \dots, \rV_3)$ are three random block-diagonal matrices with same number of diagonal blocks, where random matrices $\rV_1$, $\rV_2$, $\rV_3$ satisfies the same conditions as in Lemma~\ref{lemma:phase}. Suppose that $\rx$ is independent random vectors and $\ru_1$, $\ru_2$, and $\ru_3$ are three white CSCG vectors with unit covariance matrices and compatible size. Then 
\begin{align*}
    \frac{1}{K_1}\MI{\rudV_1^\dag \rx+ \ru_1; \rx \Big| \rudV_3^\dag \rx + \ru_3, \rudV} \geq \frac{1}{K_2}\MI{\rudV_2^\dag \rx+ \ru_2; \rx \Big| \rudV_3^\dag \rx + \ru_3, \rudV} . 
  \end{align*}
Furthermore, suppose $(\rudV_1[i], \rudV_2[i], \rudV_3[i])_{i=1}^n$ is
i.i.d.~following the joint distribution of $(\rudV_1, \rudV_2, \rudV_3)$, then
\begin{align*}
  \frac{1}{K_1}\MI{ \{ \rudV_1^\dag \rx+ \ru_1 \}^n ; \rx^n \Big| \{ \rudV_3^\dag \rx+ \ru_3 \}^n, \rudV^n} \geq \frac{1}{K_2}\MI{ \{ \rudV_2^\dag \rx+ \ru_2 \}^n; \rx^n \Big| \{ \rudV_3^\dag \rx+ \ru_3 \}^n, \rudV^n}. 
\end{align*} 
 \end{corollary}

The following result is proved in Appendix~\ref{app:maxentropy}. 

\begin{lemma}\label{lm:maxentropy}
Consider following two channels with $M$-vector input $\rx$ and fading matrices $\rA$ and $\rB$ 
\begin{subequations}\label{eq:lm:max:model}
  \begin{align}
    \ry &= \rA \rx + \rn_1 \\
    \rz &= \rB \rx + \rn_2
  \end{align}
\end{subequations}
where $\rn_1\sim \CN{0}{\mSig_1}$ and $\rn_2 \sim \CN{0}{\mSig_2}$ are mutually independent CSGC
noise, and matrix
$\begin{bmatrix} A \\ B \end{bmatrix}$
is isotropic. We also assume that
$\MEXP\|\rx\|^2 \leq \PP$. Let $\ry_G$ and $\rz_G$ be the corresponding
outputs of model~\eqref{eq:lm:max:model} with input $\rx_G\sim \CN{0}{\frac{\PP}{M}\mI_M}$,
respectively.  Then
\begin{align}
  \MI{\ry ; \rx \big| \rz, \rA, \rB} & \leq \MI{\ry_G ; \rx_G \big| \rz_G, \rA, \rB} \label{eq:maxentropy3} \\
&= \MEXP \LOG{
    \det \left(
      \begin{bmatrix}
        \mSig_1 & 0 \\
        0 & \mSig_2
      \end{bmatrix}
+ \frac{\PP}{M} 
      \begin{bmatrix}
        \rA \\ \rB
      \end{bmatrix}
      \begin{bmatrix}
        \rA^\dag & \rB^\dag
      \end{bmatrix}
\right) } \nonumber\\
& \hspace{10em}-
\MEXP\LOG{\det \left( \mSig_2 + \frac{\PP}{M}\rB\rB^\dag \right) \det \mSig_1}.
\label{eq:maxentropy1}
\end{align}
Furthermore, if conditioned on $\rx^n$, $(\ry[i], \rz[i], \rA[i], \rB[i])_{i=1}^n$ are
i.i.d.~following the joint distribution of $(\ry, \rz, \rA, \rB)$
conditioned on $\rx$, then
\begin{align}
  \MI{\ry^n ; \rx^n | \rz^n, \rA^n, \rB^n} & \leq n\MI{\ry_G ; \rx_G |
    \rz_G, \rA, \rB}\,.
\label{eq:maxentropy2}
\end{align}
\end{lemma}

\subsection{Proof of the Converse of Theorem~\ref{thm:new} with $T=1$}\label{sec:proof_con}

We prove the converse part of Theorem~\ref{thm:new} in the case of
$T=1$ in this subsection.  
The case for general $T$ will be proved in Section~\ref{sec:generalT}.
Recall that in the channel model described in Section~\ref{sec:model},
each receiver knows only the CSI of its own incoming links.  As far as
the converse proof is concerned, we assume both receivers are provided
the CSI of all links, which can only enlarge the capacity region.

The outer bounds~\eqref{eq:new1} are trivial single-user bounds.  We
establish~\eqref{eq:new2} next.

At receiver~1, by Fano's inequality and Theorem~\ref{thm:Gau}, we have
\begin{align}
  nR_1 - \delta_n
  &\leq \MI{\ry^n; \rw^n | \rH^n}. \label{eq:FanoR1t1} \\
  &\leq \MI{\rty^n; \rtw^n | \rH^n} + nC^*
\end{align}
where $\rtw[1], \dots \rtw[n]$ denote i.i.d.~white CSCG inputs, 
$\rty$ is given by~\eqref{eq:Gau4} and $C^*$ is given in~\eqref{eq:def_Cp}.
By two different uses of the chain rule on $\MI{\rty^n ; \rx^n, \rtw^n |
  \rH^n}$, we have
\begin{align}
  \MI{\rty^n ; \rtw^n | \rH^n}
  &=\MI{\rty^n ; \rx^n | \rH^n} + \MI{\rty^n ; \rtw^n | \rx^n , \rH^n}  \nonumber\\
  &\qquad- \MI{\rty^n ; \rx^n | \rtw^n, \rH^n}     
\label{eq:CRt1}
\end{align}
where two of the terms can be further simplified:
\begin{align}
  \MI{\rty^n ; \rtw^n | \rx^n , \rH^n} & = \MI{\{\rH_{11}\rtw+\ru_1\}^n ; \rtw^n | \rH^n} \\
 & = n \MEXP \log\det\left(\mI + \frac{\PP}{M_1}\rH_{11}\rH_{11}^\dag\right) \label{eq:CRt3}
\end{align}
and
\begin{align}
 \MI{\rty^n ; \rx^n | \rtw^n, \rH^n}  & = \MI{\{\rH_{12}\rx+\ru_1\}^n
   ; \rx^n | \rH^n}\,.
 \label{eq:IHx}
\end{align}
For every $r,t=1,2$, we have compact SVD
$\rH_{rt}=\rW_{rt}\rLam_{rt}\rV_{rt}^\dag$ as described in
Section~\ref{s:fading}, where $\rW_{rt}$ and $\rV_{rt}$ 
consist of orthonormal columns.  We can write
\begin{align}
  \MI{ \{   \rH_{12}\rx+\ru_1\}^n; \rx^n| \rH^n} & = \MI{ \{   \rW_{12}\rLam_{12}\rV_{12}^\dag \rx + \ru_1\}^n ; \rx^n \Big| \rH^n  }  \nonumber \\
& = \MI{ \{   \rLam_{12}\rV_{12}^\dag \rx + \rv_1\}^n ; \rx^n \Big| \rH^n } \label{eq:R1tt2} \\
& \geq \MI{ \{   \rV_{12}^\dag \rx + \rv_1\}^n ; \rx^n \Big|  \rH^n } - n\Delta_1 \label{eq:R1t2}
\end{align}
by Lemma~\ref{lemma:amp}, where $\rv_1=\rW_{12}^\dag\ru_1 \sim
\mathcal{CN}$$(0,$ $\mI_{\min(M_2,  N_1)})$,
\begin{align*} 
\Delta_1 = 2\MEXP\left[ \log^+ \frac{1}{\det(\min(\mI, \rLam_{12}))}\right]
\end{align*}
and~\eqref{eq:R1tt2} is due to the fact that given $\rH_{12}$,
$\rLam_{12}\rV_{12}^\dag\rx+\rv_1$ is a sufficient statistics of
$\rH_{12}\rx+\ru_1$ for $\rx$ (see, {\it e.g.},~\cite[Appendix
A]{Tse05book}). 
Collecting the preceding bounds, we have an upper bound on the rate of user~1: 
\begin{align}
   nR_1 & - \delta_n - nC^* \nonumber \\
            & \leq n \MI{\rty^n; \rx^n | \rH^n} +
            n\MEXP \log{\det\left(\mI + \frac{\PP}{M_1}\rH_{11}\rH_{11}^\dag\right)}            \nonumber\\
            &\qquad- \MI{ \{   \rV_{12}^\dag \rx + \rv_1\}^n ; \rx^n \Big|  \rH^n }  +n\Delta_1\,.  \label{eq:R1t4}
\end{align}

An upper bound on the rate of user~2 is obtained by Fano's
inequality and the fact that $\markov{\rx}{\rH_{22}\rx + \ru_2}{\rz}$
is Markovian:
\begin{align}
  nR_2 -\delta_n & \leq \MI{ \rz^n ; \rx^n | \rH^n} \nonumber \\
  & \leq \MI{ \{   \rH_{22}\rx + \ru_2 \}^n ; \rx^n | \rH^n } \nonumber \\
&  \leq \MI{ \{   \rV_{22}^\dag \rx + \rv_2 \}^n ; \rx^n | \rH^n } + n\Delta_2 \label{eq:R2t1}
\end{align}
where~\eqref{eq:R2t1} is by Lemma~\ref{lemma:amp} with 
\begin{align*}
\Delta_2 = 2\MEXP\log^+\det\rLam_{22} + 2\MEXP\left[ \log^+
 \frac{1}{\det(\min(\mI, \rLam_{22}))}\right]
\end{align*}
and $\rv_2=\rW_{22}^\dag \ru_2$ $\sim$ $\mathcal{CN}(0,$ $\mI_{\min(M_2, N_2)})$. 

The remaining discussion is on the two bounds~\eqref{eq:R1t4}
and~\eqref{eq:R2t1}.  In view of the three cases introduced in the
achievability proof of Theorem~\ref{thm:new}: Cases (a) $M_2\le N_1$,
(b) $M_2>N_1$ and $M_1\geq N_1$, and (c) $M_2 > N_1 > M_1$, we divide
the remaining proof of the converse by two parts: The first part
investigates Cases (a) and (b) together, and the second part
investigates Case (c).

\subsubsection{Proof of Cases (a) and (b)}

In both cases, the outer bound~\eqref{eq:new2} can be written as
\begin{align}
  d_1 + \frac{\min(M_2, N_1)}{\min(M_2,N_2)}d_2 \leq \min(M_1+M_2, N_1).\label{eq:old2}
\end{align}
We give a proof of~\eqref{eq:old2} which is similar to but
much simpler than that in~\cite{ZhuGuo09Allerton}. 

The mutual information $\MI{\rty^n , \rx^n | \rH^n}$ is that of an
isotropic fading channel with no CSIT, which is maximized by
i.i.d.~Gaussian inputs: 
\begin{align}
  \MI{\rty^n , \rx^n | \rH^n} \leq n \MEXP \LOG{\frac{\det(\mI + \frac{\PP}{M_1}\rH_{11}\rH_{11}^\dag + \frac{\PP}{M_2}\rH_{12}\rH_{12}^\dag)}{\det(\mI + \frac{\PP}{M_1}\rH_{11}\rH_{11}^\dag)}}.
\end{align}
Therefore, by~\eqref{eq:R1t4}, 
\begin{align}
   nR_1 & - \delta_n - nC^* -  n\Delta_1 \nonumber \\
& \leq n \MEXP \LOG{\det(\mI + \frac{\PP}{M_1}\rH_{11}\rH_{11}^\dag + \frac{\PP}{M_2}\rH_{12}\rH_{12}^\dag)} - \MI{ \{   \rV_{12}^\dag \rx + \rv_1\}^n ; \rx^n \Big|  \rH^n }\,.\label{eq:R1Case12t5}
\end{align}

The remaining task is to determine the ratio between the two remaining
mutual information terms in~\eqref{eq:R1Case12t5}
and~\eqref{eq:R2t1}. By noting that $\rV_{22}$ is of $M_2 \times
\min(M_2, N_2)$ and $\rV_{12}$ is of $M_2 \times \min(M_2, N_1)$ and
applying Lemma~\ref{lemma:phase}, we have 
\begin{align}
  \MI{ \{   \rV_{12}^\dag \rx + \rv_1\}^n ; \rx^n \Big|  \rH^n } \geq \frac{\min(M_2, N_1)}{\min(M_2, N_2)}\MI{ \{   \rV_{22}^\dag \rx + \rv_2\}^n ; \rx^n \Big|  \rH^n }. \label{eq:compareCase12}
\end{align}
Comparing~\eqref{eq:R2t1},~\eqref{eq:R1Case12t5}
and~\eqref{eq:compareCase12} and sending $n\to\infty$, we establish 
\begin{align} \label{eq:Case12}
  R_1 + \frac{\min(M_2, N_1)}{\min(M_2, N_2)} R_2 - \Delta \leq  \MEXP \LOG{\det(\mI + \frac{\PP}{M_1}\rH_{11}\rH_{11}^\dag + \frac{\PP}{M_2}\rH_{12}\rH_{12}^\dag)} 
\end{align}
where
\begin{align}
\Delta=C^*+\Delta_1+\frac{\min(M_2, N_1)}{\min(M_2, N_2)}\Delta_2\,.  
\end{align}
The right hand side of~\eqref{eq:Case12} is the sum ergodic capacity
of the MAC formed by the two transmitters and receiver~1. 
In the high SNR regime~($\PP \to \infty$), we have
\begin{align*}
  \MEXP \LOG{\det(\mI + \frac{\PP}{M_1}\rH_{11}\rH_{11}^\dag + \frac{\PP}{M_2}\rH_{12}\rH_{12}^\dag)}  = \min(M_1+M_2, N_1) \log\PP + o(\log\PP). 
\end{align*}
Hence~\eqref{eq:old2} is established. 

\subsubsection{Proof of Case (c)}

In Case (c), $M_2 > N_1 > M_1$,~\eqref{eq:new2} becomes 
\begin{align}\label{eq:new2c}
d_1 + \mu (d_2- L)  \leq M_1
\end{align}
where $L= N_1-M_1$ and
\begin{align}
\label{eq:mu}
\mu = \frac{M_1}{\min(M_2,N_2)-L}. 
\end{align}

To establish~\eqref{eq:new2c}, we shall use some alignment techniques
developed in~\cite{Zhu09IFC}.  We first note that the capacity region
of an interference channel depends only on the marginal distributions
of the two received signals $\ry$ and $\rz$ conditioned on the inputs,
and is otherwise invariant of the joint distribution of the outputs.
Without changing the marginals of the outputs, we assume the following
alignment in the channels and noise processes between the two users:
Let $\rV_{12}(M_2\times N_1)$ consist of the last $N_1$ columns of
$\rV_{22}(M_2\times \min(M_2,N_2))$.  Let also $\rv_1=\rW_{12}^\dag\ru_1$
consist of the last $N_1$ elements of $\rv_2=\rW_{22}^\dag\ru_2$
(both are i.i.d.~Gaussian noise).
It is important to note that $\rW_{12}$ is $N_1 \times N_1$ and unitary in this case. 

Let 
\begin{align}
  \rby = \rV_{12}^\dag \rx +   \rW_{12}^\dag\rH_{11}\rtw + \rv_1\,. \label{eq:rby}
\end{align}
We can upper bound $\MI{\rty^n; \rx^n | \rH^n}$ in~\eqref{eq:R1t4} as
follows:\footnote{This hinges on the crucial fact that $\rW_{12}$ is
  invertible in Case (c).  Because the interference plus noise, 
 $\rH_{11}\rtw+\ru_1$, is not white,  the
 equality~\eqref{eq:R1Case3t1} does not hold in general if $\rW_{12}$
 is column-rank-deficient.}
\begin{align}
  \MI{\rty^n; \rx^n | \rH^n} 
  &= \MI{\{\rW_{12}^\dag\rty\}^n; \rx^n | \rH^n} \\
& = \MI{ \{   \rLam_{12}\rV_{12}^\dag \rx +   \rW_{12}^\dag\rH_{11}\rtw + \rv_1 \}^n ; \rx^n \Big| \rH^n }  \label{eq:R1Case3t1} \\
& \leq \MI{ \rby^n ; \rx^n \Big| \rH^n } + n\Delta_3 \label{eq:R1Case3t2}
\end{align}
where~\eqref{eq:R1Case3t2} is due to Lemma~\ref{lemma:amp} and
\begin{align*}
\Delta_3 = 2\MEXP\log^+\det\rLam_{12} + 2\MEXP\left[ \log^+
 \frac{1}{\det(\min(\mI, \rLam_{12}))}\right]\,.
\end{align*}
Substituting~\eqref{eq:R1Case3t2} into~\eqref{eq:R1t4} and noting that $\markov{\rx}{  \rV_{12}^\dag \rx + \rv_1}{\rby}$ is Markovian, we can upper bound the rate of user~1 further:
\begin{align}
   nR_1 & - \delta_n - nC^* - n\Delta_1 - n\Delta_3\nonumber \\
           & \leq n\MEXP \LOG{\det(\mI + \frac{\PP}{M_1}\rH_{11}\rH_{11}^\dag)} - \MI{ \{   \rV_{12}^\dag \rx + \rv_1\}^n ; \rx^n \Big|  \rH^n } + \MI{\rby^n; \rx^n | \rH^n} \nonumber \\
            & = n \MEXP \LOG{\det(\mI + \frac{\PP}{M_1}\rH_{11}\rH_{11}^\dag)} - \MI{ \{   \rV_{12}^\dag \rx + \rv_1\}^n ; \rx^n \Big|  \rby^n, \rH^n }. \label{eq:R1Case3t3}
\end{align}

We can upper bound the rate of user~2 further by providing $\rby$ as side information in~\eqref{eq:R2t1}:
\begin{align}
  nR_2 -\delta_n - n\Delta_2 &\leq \MI{ \{   \rV_{22}^\dag \rx + \rv_2 \}^n, \rby^n ; \rx^n | \rH^n } \nonumber \\
 & = \MI{\rby^n ; \rx^n | \rH^n} + \MI{ \{   \rV_{22}^\dag \rx + \rv_2
   \}^n ; \rx^n \Big| \rby^n, \rH^n } \label{eq:R2case3t1}
\end{align}
where~\eqref{eq:R2case3t1} is due to the chain rule.

In order to establish~\eqref{eq:new2c}, we need to identify the ratio
between the last mutual information terms in~\eqref{eq:R1Case3t3} and~\eqref{eq:R2case3t1}, namely, $\MI{ \{   \rV_{12}^\dag \rx + \rv_1\}^n ; \rx^n \Big|  \rby^n, \rH^n }$ and $\mathcal{I}\Big( \{   \rV_{22}^\dag \rx + \rv_2 \}^n ; \rx^n \Big| \rby^n, \rH^n \Big)$. 
They can roughly be interpreted as the rate loss of user~1 due to
interference and the rate gain of user~2 by causing interference to user~1, respectively. 

Suppose that we have the following result (to be proved shortly):
\begin{lemma}\label{lemma:muI}
  Let $\mu$ be given by~\eqref{eq:mu}.  As $\PP\to\infty$,
  \begin{align}
    \mu \MI{ \{   \rV_{22}^\dag \rx + \rv_2 \}^n ; \rx^n \Big| \rby^n,
      \rH^n } - \MI{ \{   \rV_{12}^\dag \rx + \rv_1\}^n ; \rx^n \Big|
      \rby^n, \rH^n } \leq  n \times o(\log\PP) \label{eq:newkey} 
  \end{align}
  where the variables are as defined in this section.
\end{lemma}

Comparing~\eqref{eq:newkey} with~\eqref{eq:R1Case3t3} and~\eqref{eq:R2case3t1} and sending $n\to \infty$, we have
\begin{align}
  R_1 + &\mu R_2 - (1+\mu) \delta_n - \Delta - o(\log\PP) \nonumber \\
& \leq \MEXP \LOG{\det(\mI + \frac{\PP}{M_1}\rH_{11}\rH_{11}^\dag)} +
\frac{\mu}n \MI{\rby^n ; \rx^n | \rH^n}  \\
& \leq \MEXP \LOG{\det(\mI + \frac{\PP}{M_1}\rH_{11}\rH_{11}^\dag)} + \mu \MEXP \LOG{\frac{\det(\mI + \frac{\PP}{M_2} \rW_{12}\rW_{12}^\dag + \frac{\PP}{M_1}\rH_{11}\rH_{11}^\dag)}{\det(\mI + \frac{\PP}{M_1}\rH_{11}\rH_{11}^\dag)}} \label{eq:Rates}
\end{align}
where $\Delta=C^*-\Delta_1-\Delta_3-\mu\Delta_2$
and~\eqref{eq:Rates} is due to the fact that the mutual information
$\MI{\rby^n; \rx^n | \rH^n}$ is maximized by i.i.d.\ CSGC inputs.
Consider the approximation in the high-SNR regime~\cite{Tse05book}: 
\begin{align*}
  \MEXP \LOG{\det(\mI + \frac{\PP}{M_1}\rH_{11}\rH_{11}^\dag)} & = \min(M_1, N_1) \log\PP + o(\log \PP) \\
 \MEXP \LOG{\det(\mI + \frac{\PP}{M_2}\rW_{12}\rW_{12}^\dag + \frac{\PP}{M_1}\rH_{11}\rH_{11}^\dag)} & = \min(M_1+M_2, N_1) \log \PP + o(\log \PP)\,.
\end{align*}
Dividing both sides of~\eqref{eq:Rates} by $\LOG{1+\PP}$ and letting $\PP \to \infty$, we obtain
\begin{align*}
  d_1 + \mu d_2 &\leq \min(M_1, N_1) + \mu \big[ \min(M_1+M_2, N_1) - \min(M_1, N_1) \big] 
\end{align*}
which reduces to~\eqref{eq:new2c} under the assumption of $M_2 > N_1 > M_1$.

The remaining task is to verify that~\eqref{eq:newkey} holds.  

\begin{IEEEproof}[Proof of Lemma~\ref{lemma:muI}]
By noting that
$\rx$---$\rV_{22}^\dag\rx+\rv_2$---$\rV_{21}^\dag\rx+\rv_1$---$\rby$
is a Markov chain (due to the alignment), we have
\begin{align}
  \mu &\MI{ \{ \rV_{22}^\dag \rx + \rv_2\}^n ; \rx^n \Big| \rby^n, \rH^n } - \MI{ \{\rV_{12}^\dag \rx + \rv_1\}^n ; \rx^n \Big|  \rby^n, \rH^n }   \nonumber \\
 & = \mu  \MI{ \{\rV_{22}^\dag \rx + \rv_2\}^n  ; \rx^n \Big| \rH } - \MI{ \{\rV_{12}^\dag \rx + \rv_1\}^n ; \rx^n \Big| \rH^n} + (1-\mu) \MI{\rby^n ; \rx^n \Big| \rH^n}. \label{eq:newkeyT0}
\end{align}

 Intuitively, the interference in signal $\rby$ caused by $\rH_{11}\rtw$ is much stronger than noise in high SNR regime. However, since $N_1 > M_1$, the interference $\rH_{11}\rtw$ only occupies an $M_1$-dimension subspace. We want to show that this subspace, which contributes no DoF, can be isolated from the $N_1$-dimension received signal space so that the remaining $(N_1-M_1)$-dimension subspace can be used by user~2 without interference.   

Conditioned on $\rH$, $\rH_{11}\rtw\sim
\CN{0}{\frac{\PP}{M_1}\rH_{11}\rH_{11}^\dag}$  in~\eqref{eq:rby}
is a Gaussian random vector.
Consider the compact SVD $\rH_{11}=\rW_{11}\rLam_{11}\rV_{11}^\dag$,
where $\rLam_{11}$ is an $M_1 \times M_1$ diagonal matrix,
whose diagonal elements are strictly positive with probability 1. 
We can append orthogonal columns to $\rW_{11}$ to form a unitary
matrix $\rW = [\rW_{11}, \widetilde{\rW}_{11}]$.  Evidently,
the term $\rH_{11}\rtw$ in~\eqref{eq:rby} can be rewritten as
\begin{align}  \label{eq:2}
 \rH_{11}\rtw =
\rW
\begin{bmatrix}
\rLam_{11} \\
0
\end{bmatrix}
\rV_{11}^\dag \rtw\,.
\end{align}
Let us define
\begin{align}  \label{eq:4}
  \rtV_{12}^\dag &= \rW^\dag\rW_{12}\rV_{12}^\dag \\
  \rtv_1 &= \rW^\dag\rW_{12} \rv_1
\end{align}
where $\rtv_1\sim\CN{0}{\mI}$ is independent of $(\rW,\rW_{12})$.
Furthermore, the $N_1\times M_2$ matrix $\rtV_{12}$ can be expressed in terms of its
sub-matrices as
$\rtV_{12} = \left[ \rtV_{12,L} \,,\, \rtV_{12,R} \right]$, where
$\rtV_{12,L}$ consists of first $M_1$ columns and $\rtV_{12,R}$ consists
of the remaining $N_1-M_1$ columns. Also, let $\rtv_{1,u}$ consist of the
first $M_1$ elements in $\rtv_1$ and $\rtv_{1,d}$ consist of the remaining
$N_1-M_1$ elements.
We have
\begin{align}
   \MI{\rby^n ; \rx^n \Big| \rH^n} 
  &= \MI{ \{\rW^\dag\rW_{12} \rby\}^n ; \rx^n \Big| \rH^n}  \\
  &= \MI{ \Bigg \{ \rtV_{12}^\dag\rx + \rtv_1 + 
    \begin{bmatrix}
      \rLam_{11} \rV_{11}^\dag \rtw\\
      0
    \end{bmatrix}
     \Bigg\}^n; \rx^n \Bigg| \rH^n} \\
   &=\MI{ \{\rtV_{12,L}^\dag\rx + \rtv_{1,u} +
     \rLam_{11}\rV_{11}^\dag\rtw\}^n , \{\rtV_{12,R}^\dag\rx +
     \rtv_{1,d} \}^n; \rx^n \Big| \rH^n } \nonumber \\ 
   & = \MI{ \{ \rtV_{12,R}^\dag\rx + \rtv_{1,d} \}^n; \rx^n \Big| \rH^n} \nonumber \\
   & \qquad + \MI{ \{ \rtV_{12,L}^\dag\rx + \rtv_{1,u} +
    \rLam_{11}\rV_{11}^\dag\rtw \}^n ; \rx^n \Big| \{ \rtV_{12,R}^\dag\rx +
     \rtv_{1,d} \}^n, \rH^n} \label{eq:newkeyT3}
 \end{align}
where~\eqref{eq:newkeyT3} is due to the chain rule.
We next invoke Lemma~\ref{lm:maxentropy} on the conditional mutual
information in~\eqref{eq:newkeyT3} with $\rA=\rtV_{12,L}^\dag$,
$\rB=\rtV_{12,R}^\dag$, and the noise covariance matrices
\begin{align*}
  \rSig_1=\cov{\rtv_{1,u} + \rLam_{11}\rV_{11}^\dag\rtw} = \mI +
  \frac{\gamma}{M} \rLam_{11}^2
\end{align*}
and $\rSig_2=\mI$.  As a result,~\eqref{eq:newkeyT3} is upper bounded:
\begin{align}
  \MI{\rby^n ; \rx^n \Big| \rH^n} 
 & \leq \MI{ \{ \rtV_{12,R}^\dag\rx + \rtv_{1,d} \}^n ; \rx^n \Big|
    \rH^n} - n\MEXP \LOG{ \det \left( \frac{\PP}{M}\mI+\mI \right)
    \det\left( \mI + \frac{\PP}{M}\rLam^2_{11} \right)} \nonumber \\ 
  & \qquad + n\MEXP\LOG{\det\left(\frac{\PP}{M}\rLam^2_{11}+\mI +
      \frac{\PP}{M}\mI \right)\det \left( \frac{\PP}{M}\mI+\mI
    \right)}\label{eq:newkeyT4} \\ 
& = \MI{ \{ \rtV_{12,R}^\dag\rx + \rtv_{1,d} \}^n; \rx^n \Big| \rH^n} + 
n\MEXP\log{\det\left( \mI + \left( \rLam^2_{11} +\frac{M}{\PP} \mI \right)^{-1}\right)} \nonumber \\
& = \MI{ \{ \rtV_{12,R}^\dag\rx + \rtv_{1,d} \}^n; \rx^n \Big| \rH^n} + n\times o(\log\PP)\,. \label{eq:newkeyT5}
\end{align} 
Let us also define $\rV_{12} = [ \rV_{12,L}, \rV_{12,R} ]$ where
$\rV_{12,L}$ consists of the first $M_1$ columns.  Then $\rtV_{12,R}$
and $\rV_{12,R}$ are identically distributed.  The upper
bound~\eqref{eq:newkeyT5} can thus be rewritten as
\begin{align}
  \MI{\rby^n ; \rx^n \Big| \rH^n} 
  \le \MI{ \{ \rV_{12,R}^\dag\rx + \rv_{1,d} \}^n; \rx^n \Big| \rH^n} + n\times o(\log\PP)\,. \label{eq:newkeyT52}
\end{align}
where $\rv_{1,d}$ consists of the first $M_1$ elements of $\rv_1$ and
is identically distributed as $\rtv_{1,d}$.

Substituting~\eqref{eq:newkeyT52} into~\eqref{eq:newkeyT0}, it
suffices to show the following inequality in order to
establish~\eqref{eq:newkey}:
\begin{multline}
 \mu  \MI{ \{ \rV_{22}^\dag \rx + \rv_2\}^n  ; \rx^n \Big| \rH^n} - \MI{ \{ \rV_{12}^\dag \rx + \rv_1\}^n ; \rx^n \Big| \rH^n} \\
 + (1-\mu) \MI{ \{ \rV_{12,R}^\dag\rx + \rv_{1,d} \}^n; \rx^n \Big| \rH^n} 
\leq 0. \label{eq:newkeyTT0}
\end{multline}
Recall that $\rV_{12}(M_2\times N_1)$ consists of the last $N_1$
columns of $\rV_{22}(M_2\times\min(M_2,N_2)$ due to the assumed alignment. 
Hence $\rV_{22}$ contains all the $N_1-M_1$ columns of $\rV_{12,R}$
and we can write $\rV_{22} = [\rV_{22, L}\, \rV_{12,R}]$, where
$\rV_{22,L}$ consists of the first $p =\min(M_2, N_2)-(N_1-M_1)$
columns of $\rV_{22}$.  

Furthermore, the first $p$ elements in $\rv_{2}$ as $\rv_{2,u}$. The
remaining part of $\rv_{2}$ is $\rv_{1,d}$ due to the alignment assumption. 
Therefore, the left hand side of~\eqref{eq:newkeyTT0} is equal to 
\begin{align}
 & \mu \MI{ \{ \rV_{22}^\dag \rx + \rv_2\}^n  ; \rx^n \Big| \{ \rV_{12,R}^\dag\rx + \rv_{1,d} \}^n, \rH^n} - \MI{ \{ \rV_{12}^\dag \rx + \rv_1\}^n ; \rx^n \Big| \{ \rV_{12,R}^\dag\rx + \rv_{1,d} \}^n, \rH^n} \nonumber \\
& = \mu \MI{ \{ \rV_{22,L}^\dag \rx + \rv_{2,u} \}^n ; \rx^n \Big| \{\rV_{12,R}^\dag\rx + \rv_{1,d}\}^n , \rH^n} \nonumber \\
& \hspace{5em}- \MI{ \{ \rV_{12,L}^\dag \rx + \rv_{1,u}\}^n ; \rx^n \Big| \{\rV_{12,R}^\dag\rx + \rv_{1,d} \}^n, \rH^n}. \label{eq:newkeyTTT0}
\end{align}
Note that $\mu=M_1/p$ and $(\rV_{12,L}, \rV_{22,L}, \rV_{12,R})$ satisfy the conditions of $(\rV_1, \rV_2, \rV_3)$ in Lemma~\ref{lemma:phase}. That is, conditioned on $\rV_{12,R}$, the matrices $\rV_{12,L}$ and $\rV_{22, L}$ are uniformly distributed in the respective subspaces orthogonal to $\rV_{12,R}$.
Therefore,~\eqref{eq:newkeyTT0} follows by applying
Lemma~\ref{lemma:phase} to~\eqref{eq:newkeyTTT0}.
Thus~\eqref{eq:newkey} is established and so is
Theorem~\ref{thm:new}. 
\end{IEEEproof}

\subsection{Proof of the Converse of Theorem~\ref{thm:new} with general $T$}\label{sec:generalT}

The proof of the general case with coherence time $T$ is similar to
that of the special i.i.d.\ case ($T=1$).  Without loss of generality,
we consider the time period from $1$ to $nT$.  By stacking the
transmitted signals and noise terms at time slots
$i=(j-1)T+1,\dots,jT$ into longer vectors $\rudw[j]$, $\rudx[j]$,
$\rudu_1[j]$, and $\rudu_2[j]$, respectively, for 
$j=1,\dots, n$, The model~\eqref{eq:sys} with coherent time $T$ can be
rewritten as 
\begin{subequations}\label{eq:sys2}
  \begin{align}
    \rudy[j] &=   \rudH_{11}[j]\rudw[j] +   \rudH_{12}[j]\rudx[j] + \rudu_1[j] \\
    \rudz[j] &=   \rudH_{21}[i]\rudw[j] +   \rudH_{22}[j]\rudx[j] + \rudu_2[j] 
  \end{align}
\end{subequations}
for $j=1,\dots,n$, where for every $(r,t,j)$, $\rudH_{rt}[j]$ is an
independent block diagonal matrix with
identical diagonal blocks, i.e., $\rudH_{rt}[j]
=\text{diag}(\rH_{rt}[jT], \dots, \rH_{rt}[jT])$.

Therefore, the general case can be shown by using the equivalent channel~\eqref{eq:sys2} and following the exact same steps of the proof for case of $T=1$, where application of Lemmas~\ref{lemma:decomp} and~\ref{lemma:phase} should be replaced by the corresponding corollaries~\ref{c:decomp} and~\ref{c:phase}. The DoF region turns out to be identical as that of the case of $T=1$.

\section{Concluding Remarks}~\label{sec:con}

We have fully characterized the degree-of-freedom region of the two-user isotropic
fading MIMO interference channels without channel state information at
transmitters. In particular, we show that two users can use
independent Gaussian single-user codebooks to achieve the entire DoF
region. This suggests structured signaling schemes such as beamforming
and interference alignment cannot provide additional gains in the
high-SNR regime, although the exact capacity region remains open.  

Our result only applies to two-user interference channels with i.i.d.\
block fading, where the physical links have the same coherent time and
aligned coherence blocks. Without CSI at transmitters, interference
alignment might still provide additional gain beyond this particular
channel model.  For example, in~\cite{Jafar09BlindIA}, the author
shows that for channel with antenna configuration $(M_1, N_1, M_2,
N_2)=(1,2,3,4)$, as depicted in Fig.~\ref{fig:mimo}, if the coherent
times of receiver~1's direct link and cross link are different (say, 1
and 2, respectively), the DoF pair $(1,1.5)$ can be achieved through
interference alignment, while this DoF pair is excluded from the
region developed in Theorem~\ref{thm:new}.

\appendix

\subsection{Proof of Theorem~\ref{thm:Gau}} \label{app:Gau}

The follow result is shown in~\cite{Zamir04Gauss}:
\begin{lemma}[{\cite[Lemma~1]{Zamir04Gauss}}]
  Let $(\boldsymbol{u},\boldsymbol{v},\boldsymbol{w})$ be any real- or
  discrete-valued mutually independent random variables.  Then 
  \begin{align}    \label{eq:ZE04}
    \MI{\boldsymbol{w}+\boldsymbol{v};\boldsymbol{w}} \le 
    \MI{\boldsymbol{w}+\boldsymbol{u};\boldsymbol{w}}
    + \MI{\boldsymbol{u}+\boldsymbol{v};\boldsymbol{u}}\,.
 \end{align}
\end{lemma}

Following a similar procedure as in~\cite{Zamir04Gauss}, we can show
\begin{align}
  \MI{\mH\rw + \rv ; \mH\rw} \leq \MI{\mH\rw + \mH\rtw; \mH \rw} + \MI{ \mH\rtw+\rv ;\mH\rtw } \label{eq:thmGau:t1}
\end{align}
where 
$(\rtw, \rv, \rw)$ are mutually independent complex-valued random
vectors and $\mH$ is a determined matrix.  Moreover, $\mH\rw$ is a
sufficient statistics of $\rw$ for $\mH\rw + \rv$ and $\mH\rw +
\mH\rtw$; and  $\mH\rtw$ is a sufficient statistics of $\rtw$ for
$\mH\rtw+\rv$.  Hence~\eqref{eq:thmGau:t1} is equivalent to:
\begin{align*}
  \MI{\mH\rw + \rv ; \rw} \leq \MI{\mH\rw + \mH\rtw; \rw} + \MI{ \mH\rtw+\rv ;\rtw }\,.
\end{align*}
By noting that $\MEXP\|\rw\|^2 \leq \PP$,~\eqref{eq:Gau1} is established.

In the case of $\rtw \sim \CN{0}{\frac{\PP}{M}\mI}$, we need to show that 
\begin{align*}
  C' = \sup_{\MEXP \|\ra\|^2\leq\PP} \MI{\mH\ra + \mH\rtw; \ra} = C^*
\end{align*}
where $C^*$ is given in~\eqref{eq:def_Cp}. 
Consider the (full) SVD $\mH=\mW D\mV^\dag$, where $D$ is
$N \times M$ nonnegative and diagonal matrix, and $\mW$ and $\mV$ are
$N\times N$ and $M \times M$ unitary matrix.   We have
\begin{align*}
  C' = \sup_{\MEXP \|\ra'\|^2\leq \PP } \MI{ D \ra' + D \rtw'; \ra'}\,. 
\end{align*}
where $\ra' = \rV^\dag \ra$. We observe that $\ra'\mapsto D \ra' + D \rtw'$ is
exactly $\min(M,N)$ parallel Gaussian channels with the same gains. It
is not difficult to see that
\begin{align*}
  C' &\leq \sum_{j=1}^{\min(M,N)} \LOG{1+\frac{\PP/\min(M,N)}{\PP/M}} = C^*\,.
\end{align*}
Thus,~\eqref{eq:Gau2} is established.

For channel~\eqref{eq:sys}, by stacking $\rw^n$ and
$\{\mH_{12}\rx+\ru_1 \}^n$ into two vectors of length $nM_1$ and
$nN_1$, respectively, and applying~\eqref{eq:Gau2} with channel matrix
$\text{diag}(\mH_{11}[1],\dots, \mH_{11}[n])$, we
obtain~\eqref{eq:Gau3} if $\rH^n$ is constant. Averaging over the
distribution of $\rH^n$ yields the general result~\eqref{eq:Gau3}.

\subsection{Proof of Lemma~\ref{lemma:amp}}\label{app:amp}
 Since the two sides of~\eqref{eq:amp} are expectations over the
  joint distribution of $(\rLam_1,\rLam_2)$, it suffices to show that
  for each realization of the matrices, denoted by $(\mLam_1, \mLam_2)$,
\begin{align}
   \mathcal{I}(\mLam_1&\rx + \ru;\rx) - \MI{\mLam_2\rx + \ru;\rx} \nonumber \\
   & \geq -2\LOG{ \frac{\det \mLam_2}{\det \mLam_{\min}} } \\
   & \geq -2\log^+(\det\mLam_2)-2\log^+\left(\frac{1}{\det\mLam_{\min}}\right)
   \label{eq:amp_t1} 
 \end{align}
where $\mLam_{\min}=\min(\mLam_1, \mLam_2) > 0$.

  By data process inequality~\cite[Chapter 2]{Cover91IT},
  \begin{align}
    \mathcal{I} (\mLam_1&\rx + \ru; \rx) - \MI{\mLam_2\rx + \ru;\rx} \nonumber \\
   &\geq \MI{\mLam_{\min}\rx + \ru; \rx} - \MI{\mLam_2\rx + \ru;\rx} \nonumber \\
   &= \MI{ \mLam_2\rx + \mLam_2\mLam^{-1}_{\min} \ru; \rx} - \MI{
     \mLam_2\rx + \ru;\rx } \,.
   \label{eq:I-I}
  \end{align}
  Let $\mSig_{\ru}$ be the covariance matrix of $\ru$ and $\ru'$ be an
  independent CSCG random vector with covariance 
  $\mLam_2\mLam_{\min}^{-1}\mSig_{\ru} \mLam_{\min}^{-1} \mLam_2 -
  \mSig_{\ru}$ (which is evidently positive
  semi-definite). Then~\eqref{eq:I-I} can be further written as 
   \begin{align}
    \mathcal{I} (\mLam_1&\rx + \ru; \rx) - \MI{\mLam_2\rx + \ru;\rx} \nonumber \\
   &= \MI{ \mLam_2\rx + \ru + \ru' ; \rx } - \MI{ \mLam_2\rx + \ru;\rx } \nonumber\\
    &= - \MI{\mLam_2\rx + \ru; \rx | \mLam_2\rx + \ru + \ru'
   } \label{eq:amp_t2}
  \end{align}
  where~\eqref{eq:amp_t2} is because $\markov{\rx}{\mLam_2\rx +
    \ru}{\mLam_2\rx + \ru + \ru'}$ is Markov.  Therefore, it boils
  down to upper bounding the mutual information in~\eqref{eq:amp_t2}:
 \begin{align}
    \mathcal{I}( \mLam_2\rx + &\ru; \rx | \mLam_2\rx + \ru + \ru' )
    \nonumber \\
    &=  \MI{ \ru' ; \ru + \ru' | \mLam_2\rx + \ru +
      \ru'}  \nonumber \\
    &  \leq \MI{ \ru' ; \ru + \ru' } \label{eq:amp_t4} \\
    &= 2\LOG{ \frac{\det
        \mLam_2}{\det \mLam_{\min}} } \nonumber \\
    &\leq 2\log^+\det \mLam_2 +
    2\log^+\left(\frac{1}{\det\mLam_{\min}}\right) \nonumber
  \end{align}
  where in~\eqref{eq:amp_t4} we have used the fact that
  $\markov{\ru'}{\ru + \ru'}{\mLam_2\rx + \ru + \ru'}$ forms a Markov
  chain. 
  We have thus established~\eqref{eq:amp_t1}. Lemma~\ref{lemma:amp} follows by
  taking the expectation on both sides.

\subsection{Proof of Lemma~\ref{lemma:phase}} \label{app:phase}
Let a random vector $\rx$ and another random object $\rv$ have a joint
distribution.  Define the minimum mean-square error~(MMSE) of
estimating $\rx$ conditional on $\rv$ and $\sqrt{t}\,\rx+\ru$,
where $\ru\sim\mathcal{CN}(0,\mI)$ is independent of $(\rx,\rv)$ as
  \begin{align}
    \mmse{\rx\,;t | \rv} = \MEXP\left[
      \left\| \rx-\MEXP \Big[\rx\big|\sqrt{t}\,\rx+\ru, \rv\Big] \right\|^2 \right]\,.
\end{align}
We have the following formula that relates the MMSE and mutual information~\cite{Guo11MMSE}:
\begin{align}
  \MI{\sqrt{t}\,\rx+\ru; \rx \big| \rv} = \int_0^t \mmse{ \rx; \tau | \rv }  \ud \tau \label{eq:app_max_rev1_t2}
\end{align}

  Find an arbitrary orthonormal basis in space $\mathbb{C}^{K_2}$, say, $\{e_i\}_1^{K_2}$;
  then construct $K_2$ subsets of $\{e_i\}_1^{K_2}$ such that each subset has
  $K_1$ elements and each $e_i$ is included in exact $K_1$ subsets; each subset
  corresponds to a $K_1 \times K_2$ matrix, called $\mB_1, \dots,
  \mB_{K_2}$. Then we see that $\mB_j\mB_j^\dag=\mI_{K_1}$ for all $j=1,\dots,K_2$ and $\frac{1}{K_1}\sum_{j=1}^{K_2}\mB_j^\dag \mB_j=\mI_{K_2}$. Therefore, for any $\rv$ and $\rz$
  \begin{align}
    &\frac{1}{K_1} \sum_{j=1}^{K_2} \mmse{\mB_j\rz\,; t \big| \rv}  \nonumber \\
    & = \frac{1}{K_1} \sum_{j=1}^{K_2} \MEXP\left[ \left\| \mB_j\rz -
        \MEXP\left[\mB_j\rz \big| \Sqrt{t} \mB_j\rz + \mB_j\ru_2,
          \rv\right] \right\|^2 \right] \nonumber \\
    & \geq \frac{1}{K_1} \sum_{j=1}^{K_2} \MEXP\Big[ \left( \mB_j\rz - \MEXP\left[\mB_j\rz \big| \Sqrt{t} \rz + \ru_2, \rv\right] \right)^\dag \left( \mB_j\rz - \MEXP\left[\mB_j\rz \big| \Sqrt{t} \rz + \ru_2, \rv\right] \right) \Big] \label{eq:phaseT1} \\
& = \MEXP\left[ \left( \rz - \MEXP\left[\rz \big| \Sqrt{t} \rz + \ru_2, \rv \right] \right)^\dag \left( \frac{1}{K_1} \sum_{j=1}^{K_2} \mB_j^\dag \mB_j \right) \left( \rz - \MEXP\left[\rz \big| \Sqrt{t} \rz + \ru_2, \rv\right] \right) \right] \nonumber \\
& = \MEXP\left[ \left( \rz - \MEXP\left[\rz \big| \Sqrt{t} \rz + \ru_2, \rv \right] \right)^\dag \left( \rz - \MEXP\left[\rz \big| \Sqrt{t} \rz + \ru_2, \rv \right] \right) \right] \nonumber \\
& = \mmse{\rz\,;t\big| \rv}  \label{eq:k1k2}
  \end{align}
where~\eqref{eq:phaseT1} is due to the fact that we have better estimation with better observation. 
Letting $\rz = \rV_2^\dag \rx$ and $\rv=\Big( \rV_3^\dag \rx + \ru_3,
\rV \Big)$ in~\eqref{eq:k1k2}, we have
\begin{align}
  \frac{1}{K_1} \sum_{j=1}^{K_2} \mmse{\mB_j\rV_2^\dag \rx \,; t \big| \rV_3^\dag \rx + \ru_3, \rV } \geq \mmse{\rV_2^\dag\rx \,; t \big| \rV_3^\dag \rx + \ru_3, \rV} \label{eq:phaseT2}
\end{align}

Furthermore, $\Probb_{\mB_j\rV_2|\rV_3}$ and $\Probb_{\rV_1|\rV_3}$ are uniform
distributions on $\mathcal{V}_1$ by assumption, hence
$(\mB_j\rV_2, \rV_3)$ and $(\rV_1, \rV_3)$ are identically distributed. Therefore, 
\begin{align}
  &\frac{K_2}{K_1} \MI{\rV_1^\dag \rx+ \ru_1; \rx \Big| \rV_3^\dag \rx + \ru_3, \rV} \nonumber \\
  &\;\; = \frac{1}{K_1} \sum_{j=1}^{K_2} \MI{\mB_j\rV_2^\dag \rx+ \ru_1; \rx \Big| \rV_3^\dag \rx + \ru_3, \rV} \nonumber \\
&\;\; = \frac{1}{K_1} \sum_{j=1}^{K_2} \int_0^1 \mmse{\mB_j\rV_2^\dag\rx\,; t \Big| \rV_3^\dag \rx + \ru_3, \rV} \ud t \label{eq:phaseT3} \\
&\;\; \geq \int_0^1 \mmse{\rV_2^\dag\rx\,;t\Big| \rV_3^\dag \rx + \ru_3, \rV} \ud t \label{eq:phaseT4} \\
&\;\; = \MI{\rV_2^\dag \rx+ \ru_2; \rx \Big| \rV_3^\dag \rx + \ru_3, \rV} \label{eq:phaseT5}
\end{align}
where~\eqref{eq:phaseT3} and~\eqref{eq:phaseT5} are due
to~\eqref{eq:app_max_rev1_t2}, and~\eqref{eq:phaseT4} is due to~\eqref{eq:phaseT2}. 
We have thus established~\eqref{eq:phase1}.

To show~\eqref{eq:phase2}, we stack $\rx[1], \dots, \rx[n]$ into a vector
 $\overline{\rx}$ of size $nM$, stack $\ru_j[1], \dots, \ru_j[n]$ into a
 vector $\overline{\ru}_j$ of size $nN_j$ for $j=1,2$, and construct random matrix
 $\overline{\rV}_j=\text{diag}(\rV_j[1], \dots, \rV_j[n])$ for $j=1,2$. Then the sequence $\{\rV_j^\dag[i] \rx[i] + \ru_j[i]\}_{i=1}^n$ can be represented as $\overline{\rV_j}^\dag\overline{\rx}+\overline{\ru}_j$. Let $\overline{\mB}_j=\text{diag}(\mB_j, \dots, \mB_j)$. It is easy to see that $\overline{\mB}_j \overline{\mB}_j=\mI_{nK_1}$ and $\frac{1}{K_1}\sum_{j=1}^{K_2}\overline{\mB}_j^\dag \overline{\mB}_j= \mI_{nK_2}$. Although $\overline{\rV}_j$ are not uniformly distributed, it is still true that $(\overline{\mB}_j\overline{\rV}_2, \overline{\rV}_3)$ and $(\overline{\rV}_1, \overline{\rV}_3)$ have identical distribution. 
Therefore,~\eqref{eq:phase2} follows by similar arguments as in above.

\subsection{Proof of Lemma~\ref{lm:maxentropy}}\label{app:maxentropy}

The equality~\eqref{eq:maxentropy1} is straightforward. We focus on the
  inequality~\eqref{eq:maxentropy3}. 

Consider the eigenvalue decomposition of the noise variance $\Sigma_1=W_1 \Lambda_1
 W_1^\dag$, then $\ry' = W_1\Lambda^{-1/2}\ry = \rA' \rx + \rn_1'$,
 where $\rn_1'= W_1\Lambda_1^{-1/2} \rn_1 \sim \CN{0}{I}$ and $\rA' =
 W_1\Lambda_1^{-1/2} \rA$, which is still isotropic. Also, $\ry'$ is a
 sufficient statistics of $\ry$. Therefore,
 applying~\eqref{eq:app_max_rev1_t2} with  $\rv=(\rz, \rA', \rB)$, we
 have
\begin{align}
  \MI{\ry; \rx | \rz, \rA, \rB}
  &=  \MI{\ry'; \rx | \rz, \rA', \rB} \nonumber \\
  & = \int_0^1 \mmse{ \rA'\rx; t \Big| \rz, \rA', \rB} \ud t. \label{eq:app_max_rev1_t1} \end{align}
Note that $\rA'$ is still isotropic by Definition 1.

Given $\rA'=A'$ and $\rB=B$, the MMSE in~\eqref{eq:app_max_rev1_t1}
can be expressed as
\begin{align}  \label{eq:EAx}
  \mmse{A'\rx;t|\rz}
  &= \mmse{A'\rx;t|B\rx+\rn_2} \\
  &= \MEXP 
  \left\| A'\rx - A'\MEXP\left[ \rx\Bigg|
        \begin{bmatrix}
          \sqrt{t}A' \\ B
       \end{bmatrix}
       \rx+
       \begin{bmatrix}
         \rn_1' \\ \rn_2
       \end{bmatrix}
     \right] \right\|^2 
\end{align}
which is the MMSE of $A'\rx$ conditioned on a linear transformation of
$\rx$ with additive Gaussian noise.  Let the covariance of $\rx$ be
$\mQ=\cov{\rx}$.  Let $\rx_{\mQ} \sim \CN{0}{\mQ}$ be Gaussian with the same
covariance.  Then the MMSE~\eqref{eq:EAx} cannot decrease if the input
$\rx$ is replaced by $\rx_{\mQ}$, i.e.,
\begin{align}  \label{eq:AxG}
  \mmse{ A'\rx;t | \rz } \leq \mmse{ A'\rx_{\mQ}; t | \rz_{\mQ}}
\end{align}
holds for every $t\ge0$,
where $\rz_{\mQ}=B\rx_{\mQ}+\rn_2$.  The reason is that the estimator that
minimizes the MMSE for $A'\rx_{\mQ}$ is linear, which also achieves the
same MMSE if applied to $A'\rx$.  This implies that using the optimal
(nonlinear) estimator for $A'\rx$ can only yield a smaller MMSE.

Plugging~\eqref{eq:AxG} into~\eqref{eq:app_max_rev1_t1}, we see that,
in order to maximize the mutual information $\MI{\ry; \rx | \rz, \rA,
  \rB}$, it suffices to restrict the input vector on the set of
Gaussian random vectors, i.e., it boils down to finding the 
covariance matrix $\mQ$ that maximizes the mutual information.  As we
shall see, the optimal $\mQ$ is $(\PP/M)\mI_M$.

Consider the eigenvalue decomposition $\mQ=\mU \Lambda \mU^\dag$.
Then $\mU^\dag\rx_{\mQ}$ consists of independent entries.
Due to the isotropy of $\rA'$ and $\rB$, the statistics of
$\rA'\mU^\dag\rx_{\mQ}$ and $\rB\mU^\dag\rx_{\mQ}$ 
are identically distributed as
$\rA'\rx_{\mQ}$ and $\rB\rx_{\mQ}$, respectively.  Hence the MMSE is
invariant to the eigenvectors of $\mQ$. 
Therefore, the maximization problem can be further restricted to all
Gaussian $\rx_{\mQ}$ with independent entries, i.e., $\mQ$ is diagonal. 

To maximize the mutual information, the diagonal entries of $\mQ$ must
all be equal: Let $\pi$ be the collection of all $M!$ permutation
matrices for the $M$-dimension linear space. By isotropy of $\rA$ and
the concavity of conditional MMSE, we have 
\begin{align}
  \mmse{ \rA'\rx_{\mQ} ; t \Big| \rz_{\mQ}, \rA', \rB}
  &= \frac{1}{M!}\sum_{\Pi \in \pi} \mmse{ \rA'\rx_{\Pi\mQ\Pi^\dag}; t \Big| \rz_{\Pi\mQ\Pi^\dag}, \rA', \rB} \nonumber \\
  &\leq \mmse{ \rA' \rx_{\mR} | \rz_{\mR}, \rA', \rB} \nonumber 
\end{align}
where
\begin{align}
  \label{eq:3}
  \mR = \frac{1}{M!} \sum_{\Pi \in \pi} \Pi\mQ\Pi^\dag
\end{align}
have identical diagonal entries.
Therefore, to maximize the mutual information, we can further restrict
the optimization problem to be on Gaussian i.i.d.\ inputs. In other words,  
\begin{align}
  \MI{\ry; \rx | \rz, \rA, \rB}  \leq \MI{\rA\rx_{\rho\mI}+\rn_1; \rx_{\rho\mI} | \rz_{\rho\mI}, \rA, \rB} 
\end{align}
for some $\rho\le \PP/M$.

Finally, we show that the maximum mutual information is achieved by $\rho=\PP/M$.
Suppose otherwise, i.e., $\rho < \PP/M$.  For convenience, denote
$\rx_{\rho\mI}$ by $\rx_\rho$.  Let $\rhx \sim
\CN{0}{(\PP/M-\rho)\mI_M}$ be independent of $\rx_\rho$. Then
$x_{\PP/M} = x_\rho + \hat{x}$.  Given
$\rA=\mA$ and $\rB=\mB$,
\begin{align}
  \MI{\mA\rx_{\rho}+\rn_1; \rx_{\rho} | \rz_{\rho}}
  &= \MI{\mA(\rx_{\rho}+\rhx)+\rn_1; \rx_{\rho}+\rhx | 
   \mB(\rx_{\rho}+\rhx)+\rn_2, \rhx} \nonumber \\
  & = \MI{\mA\rx_{\PP/M}+\rn_1; \rx_{\PP/M} | 
   \mB\rx_{\PP/M}+\rn_2, \rhx} \nonumber \\
& \leq  \MI{\mA\rx_{\PP/M}+\rn_1; \rx_{\PP/M}, \rhx | 
   \mB\rx_{\PP/M}+\rn_2}  \label{eq:app_max_rev1_t3}\\
 & = \MI{\mA\rx_{\PP/M}+\rn_1; \rx_{\PP/M} | 
   \mB\rx_{\PP/M}+\rn_2} + \MI{\mA\rx_{\PP/M}+\rn_1; \rhx | 
   \mB\rx_{\PP/M}+\rn_2, \rx_{\PP/M}} \nonumber \\
 & = \MI{\mA\rx_{\PP/M}+\rn_1; \rx_{\PP/M} | 
   \mB\rx_{\PP/M}+\rn_2} + \MI{ \rn_1; \rhx | \rn_2, \rx_{\PP/M}} \nonumber \\
 & = \MI{\mA\rx_{\PP/M}+\rn_1; \rx_{\PP/M} | 
   \mB\rx_{\PP/M}+\rn_2} \label{eq:79} 
\end{align}
where in~\eqref{eq:app_max_rev1_t3} is due to chain rule
and~\eqref{eq:79} is due to independence of the signals and the
noises. Similarly,~\eqref{eq:maxentropy2} can be proved by stacking
the sequences of vectors into larger vectors. 

\section*{Acknowledgment}

The authors would like to thank Associate Editor Syed Jafar for useful
suggestions and for pointing out a mistake in the proof in an earlier
draft of the paper. 

\bibliographystyle{IEEEtran}
\bibliography{IEEEabrv,zybibset,zybib_IC}
\end{document}

%% file: zymacroset.tex
\usepackage{ifthen}

\newtheorem{theorem}{Theorem}

\newtheorem{lemma}{Lemma}

\newtheorem{defn}{Definition}

\ifthenelse{\equal{\dongningstyle}{true}}{
        \newcommand{\RVEC}[1]{\boldsymbol{\uppercase{#1}}}
        \newcommand{\RMAT}[1]{\underline{\boldsymbol{\uppercase{#1}}}}

        \newcommand{\MAT}[1]{\boldsymbol{\underline{\lowercase{#1}}}}
}{
        \newcommand{\RVEC}[1]{\boldsymbol{\lowercase{#1}}}
        \newcommand{\RMAT}[1]{\boldsymbol{\boldsymbol{\uppercase{#1}}}}

        \newcommand{\MAT}[1]{\uppercase{#1}}
}

\newcommand{\mmse}[1]{\mathsf{mmse}\left( #1 \right)}

\newcommand{\MEXP}{\mathbb{E}}

\newcommand{\Probb}{\mathsf{P}}
\newcommand{\LOG}[1]{\log \left( #1 \right)}

\newcommand{\uT}{\mathsf{T}}
\newcommand{\ud}{\mathsf{d}}

\newcommand{\IE}{{\it i.e.}}

%% file: NewDoFICMIMO.bbl
\begin{thebibliography}{10}
\providecommand{\url}[1]{#1}
\csname url@samestyle\endcsname
\providecommand{\newblock}{\relax}
\providecommand{\bibinfo}[2]{#2}
\providecommand{\BIBentrySTDinterwordspacing}{\spaceskip=0pt\relax}
\providecommand{\BIBentryALTinterwordstretchfactor}{4}
\providecommand{\BIBentryALTinterwordspacing}{\spaceskip=\fontdimen2\font plus
\BIBentryALTinterwordstretchfactor\fontdimen3\font minus
  \fontdimen4\font\relax}
\providecommand{\BIBforeignlanguage}[2]{{%
\expandafter\ifx\csname l@#1\endcsname\relax
\typeout{** WARNING: IEEEtran.bst: No hyphenation pattern has been}%
\typeout{** loaded for the language `#1'. Using the pattern for}%
\typeout{** the default language instead.}%
\else
\language=\csname l@#1\endcsname
\fi
#2}}
\providecommand{\BIBdecl}{\relax}
\BIBdecl

\bibitem{Etkin08IFC}
R.~H. Etkin, D.~N.~C. Tse, and H.~Wang, ``Gaussian interference channel
  capacity to within one bit,'' \emph{{IEEE} Trans. Inf. Theory}, vol.~54,
  no.~12, pp. 5534--5562, Dec. 2008.

\bibitem{Cadambe08IFA}
V.~R. Cadambe and S.~A. Jafar, ``Interference alignment and degrees of freedom
  of the {K-user} interference channel,'' \emph{{IEEE} Trans. Inf. Theory},
  vol.~54, no.~8, pp. 3425--3441, Aug. 2008.

\bibitem{Shang10MIMOIC}
X.~Shang, B.~Chen, G.~Kramer, and H.~V. Poor, ``Capacity regions and sum-rate
  capacities of vector {Gaussian} interference channels,'' \emph{{IEEE} Trans.
  Inf. Theory}, vol.~56, no.~10, pp. 5030--5044, Oct. 2010.

\bibitem{Gou08KIFC-MIMO}
T.~Gou and S.~A. Jafar, ``Degrees of freedom of the {$K$} user {$M \times N$}
  {MIMO} interference channel,'' \emph{{IEEE} Trans. Inf. Theory}, vol.~56,
  no.~12, pp. 6040--6057, Dec 2010.

\bibitem{Annapureddy09MIMOIC}
\BIBentryALTinterwordspacing
V.~S. Annapureddy and V.~V. Veeravalli, ``Sum capacity of {MIMO} interference
  channels in the low interference regime,'' \emph{preprint}, Sep. 2009.
  [Online]. Available: \url{http://arxiv.org/abs/0909.2074v1}
\BIBentrySTDinterwordspacing

\bibitem{Jafar07MIMOIFC}
S.~A. Jafar and M.~J. Fakhereddin, ``Degrees of freedom for the {MIMO}
  interference channels,'' \emph{{IEEE} Trans. Inf. Theory}, vol.~53, no.~7,
  pp. 2637--2642, Jul. 2007.

\bibitem{RajPra09IT}
A.~Raja, V.~M. Prabhakaran, and P.~Viswanath, ``The two-user compound
  interference channel,'' \emph{{IEEE} Trans. Inf. Theory}, vol.~55, no.~11,
  pp. 5100--5120, Nov. 2009.

\bibitem{Raja09DMT-IFC}
A.~Raja and P.~Viswanath, ``Diversity-multiplexing tradeoff of the two-user
  interference channel,'' \emph{{IEEE} Trans. Inf. Theory}, 2011, to appear.

\bibitem{Akuiyibo08IFCMIMO}
E.~Akuiyibo, O.~L{\'e}v{\^e}que, and C.~Vignat, ``High {SNR} analysis of the
  {MIMO} interference channel,'' in \emph{Proc. IEEE Int. Symp. Inf. Theory},
  Toronto, Jul. 2008, pp. 905 -- 909.

\bibitem{Huang09MIMODoFb}
\BIBentryALTinterwordspacing
C.~Huang, S.~A. Jafar, S.~{Shamai~(Shitz)}, and S.~Vishwanath, ``On degrees of
  freedom region of {MIMO} networks without {CSIT},'' \emph{preprint}, 2009.
  [Online]. Available: \url{http://arxiv.org/abs/0909.4017}
\BIBentrySTDinterwordspacing

\bibitem{Vaze09IFC}
\BIBentryALTinterwordspacing
C.~S. Vaze and M.~K. Varanasi, ``The degrees of freedom regions of {MIMO}
  broadcast, interference, and cognitive radio channels with no {CSIT},''
  \emph{preprint}, Oct. 2009. [Online]. Available:
  \url{http://arxiv.org/abs/0909.5424v2}
\BIBentrySTDinterwordspacing

\bibitem{ZhuGuo09Allerton}
Y.~Zhu and D.~Guo, ``Isotropic {MIMO} interference channels without {CSIT: T}he
  loss of degrees of freedom,'' in \emph{Proc.\ Allerton Conf.\ Commun.,
  Control, and Computing}.\hskip 1em plus 0.5em minus 0.4em\relax Monticello,
  IL, USA, Oct. 2009.

\bibitem{Tse05book}
D.~N.~C. Tse and P.~Viswanath, \emph{Fundamentals of Wireless
  Communications}.\hskip 1em plus 0.5em minus 0.4em\relax Cambridge University
  Press, 2005.

\bibitem{Zheng02MIMO}
L.~Zheng and D.~Tse, ``Communicating on the {Grassmann} manifold: {A} geometric
  approach to the non-coherent multiple antenna channel,'' \emph{{IEEE} Trans.
  Inf. Theory}, vol.~48, no.~2, pp. 359--383, Feb. 2002.

\bibitem{Telatar99MIMO}
E.~Telatar, ``Capacity of multi-antenna {Gaussian} channels,'' \emph{European
  Transactions on Telecommunications}, vol.~10, no.~6, pp. 585--595, 1999.

\bibitem{Jafar09BlindIA}
\BIBentryALTinterwordspacing
S.~A. Jafar, ``Exploiting channel correlations -- {Simple} interference
  alignment schemes with no {CSIT},'' \emph{preprint}, Oct. 2009. [Online].
  Available: \url{http://arxiv.org/abs/0910.0555v1}
\BIBentrySTDinterwordspacing

\bibitem{KeWan10X}
\BIBentryALTinterwordspacing
L.~Ke and Z.~Wang, ``Degrees of freedom regions of two-user {MIMO Z} and full
  interference channels: The benefit of reconfigurable antennas,'' \emph{{IEEE}
  Trans. Inf. Theory}, Sep. 2010, submitted. [Online]. Available:
  \url{http://arxiv.org/pdf/1011.2196}
\BIBentrySTDinterwordspacing

\bibitem{Han81Interference}
T.~S. Han and K.~Kobayashi, ``A new achievable rate region for the interference
  channel,'' \emph{{IEEE} Trans. Inf. Theory}, vol.~27, no.~1, pp. 49--60, Jan
  1981.

\bibitem{Zamir04Gauss}
R.~Zamir and U.~Erez, ``A {Gaussian} input is not too bad,'' \emph{{IEEE}
  Trans. Inf. Theory}, vol.~50, no.~6, pp. 1362 -- 1367, Jun. 2004.

\bibitem{Zhu09IFC}
Y.~Zhu and D.~Guo, ``Ergodic fading {Z}-interference channels without state
  information at transmitters,'' \emph{{IEEE} Trans. Inf. Theory}, vol.~57,
  no.~5, pp. 2627 -- 2647, May 2011.

\bibitem{Cover91IT}
T.~M. Cover and J.~A. Thomas, \emph{Elements of Information Theory},
  3rd~ed.\hskip 1em plus 0.5em minus 0.4em\relax John Wiley \& Sons, Inc.,
  2006.

\bibitem{Guo11MMSE}
D.~Guo, Y.~Wu, S.~{Shamai (Shitz)}, and S.~Verd{\'u}, ``Estimation of
  {non-Gaussian} random variables in {Gaussian} noise: {Properties} of the
  minimum mean-square error,'' \emph{{IEEE} Trans. Inf. Theory}, vol.~57, April
  2011.

\end{thebibliography}
